\documentclass[12pt,draftcls, onecolumn]{IEEEtran}

\usepackage[normalem]{ulem}

\usepackage{amsmath}
\usepackage{amsthm}
\usepackage{amssymb}
\usepackage{verbatim}
\usepackage{algorithm}
\usepackage{algorithmic}
\usepackage{graphicx}
\usepackage{paralist}
\usepackage{color}
\usepackage{hyperref}
\usepackage{url}

\usepackage{xcolor,colortbl}
\usepackage{verbatim}
\usepackage{algorithm}
\usepackage{algorithmic}
\usepackage{graphicx}
\usepackage{epstopdf}  
\usepackage{paralist}
\usepackage{color}
\usepackage{subfigure}
\usepackage{ulem}

\newtheorem{lemma}{Lemma}
\newtheorem{theorem}{Theorem}
\newtheorem{definition}{Definition}

\floatname{algorithm}{Algorithm}
\newtheorem*{remark}{Remark}

\title{From Static to Dynamic Tag Population Estimation: \\ An Extended Kalman Filter Perspective}
\author{Jihong Yu, Lin Chen\thanks{J. Yu and L. Chen are with Lab. Recherche Informatique (LRI-CNRS UMR 8623), Univ. Paris-Sud, 91405 Orsay, France, \{Jihong.Yu, Lin.Chen\}@lri.fr.}}

\begin{document}

\maketitle
\begin{abstract}
Tag population estimation has recently attracted significant research attention due to its paramount importance on a variety of radio frequency identification (RFID) applications. However, most, if not all, of existing estimation mechanisms are proposed for the static case where tag population remains constant during the estimation process, thus leaving the more challenging dynamic case unaddressed, despite the fundamental importance of the latter case on both theoretical analysis and practical application. In order to bridge this gap, 
we devote this paper to designing a generic framework of stable and accurate tag population estimation schemes based on Kalman filter for both static and dynamic RFID systems.
Technically, we first model the dynamics of RFID systems as discrete stochastic processes and leverage the techniques in extended Kalman filter (EKF) and cumulative sum control chart (CUSUM) to estimate tag population for both static and dynamic systems. By employing Lyapunov drift analysis, we mathematically characterise the performance of the proposed framework in terms of estimation accuracy and convergence speed by deriving the closed-form conditions on the design parameters under which our scheme can stabilise around the real population size with bounded relative estimation error that tends to zero with exponential convergence rate.
\end{abstract}


\section{Introduction}

\subsection{Context and Motivation}
Recent years have witnessed an unprecedented development and application of the radio frequency identification (RFID) technology. As a promising low-cost technology, RFID is widely utilized in various applications ranging from inventory control~\cite{DoD2004}~\cite{DoD2007}, supply chain management~\cite{lee2008supply} to tracking/location \cite{ni2011tracking}~\cite{yang2013localization}. A standard RFID system has two types of devices: a set of RFID tags and one or multiple RFID readers (simply called \textit{tags} and \textit{readers}). A tag is typically a low-cost microchip labeled with a unique serial number (ID) to identify an object. A reader, on the other hand, is equipped with an antenna and can collect the information of tags within its coverage area.

\textit{Tag population estimation and counting} is a fundamental functionality for many RFID applications such as warehouse management, inventory control and tag identification. For example, quickly and accurately estimating the number of tagged objects is crucial in establishing inventory reports for large retailers such as Wal-Mart~\cite{Wal-Mart2005}. Moreover, as the \textit{de facto} MAC layer protocol for RFID systems, the framed-slotted ALOHA protocol~\cite{cha2006DFSA} requires the optimal frame size to be set to the number of tags in the system.


Due to the paramount practical importance of tag population estimation, a large body of studies~\cite{kodialam2007anonymous}~\cite{li2010energy}~\cite{qian2011cardinality}~\cite{shahzad2012everybit} \cite{zheng2013zoe} have been devoted to the design of efficient estimation algorithms. Most of them, as reviewed in Sec.~\ref{sec:related}, are focused on the static scenario where the tag population is constant during the estimation process. However, many practical RFID applications, such as logistic control, are dynamic in the sense that tags may be activated or terminated as specialized in C1G2 standard~\cite{C1G22005}, or they may enter and/or leave the reader's covered area, thus resulting in tag population variation. In such dynamic applications, a fundamental research question is how to design efficient algorithms to dynamically trace the tag population quickly and accurately.


\subsection{Summary of Contributions}

In this paper, we develop a generic framework of stable and accurate tag population estimation schemes for both static and dynamic RFID systems.
By generic, we mean that our framework both supports the real-time monitoring and can estimate the number of tags accurately without any prior knowledge on the tag arrival and departure patterns. Our design is based on the extended Kalman filter (EKF)~\cite{song1992EKF}, a powerful tool in optimal estimation and system control. We also use the estimated tag number to dynamically update the frame size in the framed-slotted ALOHA protocol. By performing Lyapunov drift analysis, we mathematically prove the efficiency and stability of our framework.


The main technical contributions of this paper are articulated as follows. We formulate the system dynamics of the tag population for both static and dynamic RFID systems where the number of tags remains constant and varies during the estimation process.
We design an EKF-based population estimation algorithm for static RFID systems and further enhance it to dynamic RFID systems by leveraging the cumulative sum control chart (CUSUM) to detect the population change. 
By employing Lyapunov drift analysis, we mathematically characterise the performance of the proposed framework in terms of estimation accuracy and convergence speed by deriving the closed-form conditions on the design parameters under which our scheme can stabilise around the real population size with bounded relative estimation error that tends to zero within exponential convergence rate.
To the best of our knowledge, our work is the first theoretical framework that dynamically traces the tag population with closed form conditions on the estimation stability and accuracy.

\section{Related Work}
\label{sec:related}

Due to its fundamental importance, tag population estimation has received significant research attention, which we briefly review in this section.

\subsection{Tag Population Estimation for Static RFID systems}

Most of existing works are focused on the static scenario where the tag population is constant during the estimation process. The central question there is to design efficient algorithms quickly and accurately estimating the static tag population. Kodialam \textit{et al.} design an estimator called PZE which uses the probabilistic properties of empty and collision slots to estimate the tag population size~\cite{kodialam2006mobicom}.
The authors then further enhance PZE by taking the average of the probability of idle slots in multiple frames as an estimator in order to eliminate the constant additive bias~\cite{kodialam2007anonymous}.
Han \textit{et al.} exploit the average number of idle slots before the first non-empty slots to estimate the tag population size~\cite{han2010anonymously}. 
Later, \textit{Qian et al.} develop Lottery-Frame scheme that employs geometrically distributed hash function such that the $j$th slot is chosen with prob. $\frac{1}{2^{j+1}}$~\cite{qian2011cardinality}. As a result, the first idle slot approaches around the logarithm of the tag population and the frame size can be reduced to the logarithm of the tag population, thus reducing the estimation time.
Subsequently, a new estimation scheme called ART is proposed in~\cite{shahzad2012everybit} based on the average length of consecutive non-empty slots. The design rational of ART is that the average length of consecutive non-empty slots is correlated to the tag population. ART is shown to have smaller variance than prior schemes.
More recently, \textit{Zheng et al.} propose another estimation algorithm, ZOE, where each frame just has a single slot and the random variable indicating whether a slot is idle follows Bernoulli distribution~\cite{zheng2013zoe}. The average of multiple individual observations is used to estimate the tag population.

We would like to point out that the above research work does not consider the estimation problem for dynamic RFID systems and thus may fail to monitor the system dynamics in real time. Specifically, in typical static tag population estimation schemes, the final estimation result is the average of the outputs of multi-round executions. When applied to dynamic tag population estimation, additional  estimation error occurs due to the variation of the tag population size during the estimation process.

\subsection{Tag Population Estimation for Dynamic RFID systems}

Only a few propositions have tackled the dynamic scenario.
The works in~\cite{sarangan2008mobile} and~\cite{xie2010mobile} consider specific tag mobility patterns that the tags move along the conveyor in a constant speed, while tags may move in and out by different workers from different positions, so these two algorithm cannot be applicable to generic dynamic scenarios. Subsequently, Xiao \textit{et al.} develop a differential estimation algorithm, ZDE, in dynamic RFID systems to estimate the number of arriving and removed tags~\cite{xiao2013differential}. More recently, they further generalize ZDE by taking into account the snapshots of variable frame sizes~\cite{xiao2015temporally}. Though the algorithms in~\cite{xiao2013differential} and~\cite{xiao2015temporally} can monitor the dynamic RFID systems, they may fail to estimate the tag population size accurately, because they just use the same hash seed in the whole monitoring process. Using the same seed is an effective way in tracing tag departure and arrival. However, using the same seed may significantly limit the estimation accuracy, even in the static case.

Besides the limitations above, prior works do not provide formal analysis on the stability and the convergence rate. To full this vide, we develop a generic framework for tag population estimation in dynamic RFID systems. By generic, we mean that our framework can both support real-time monitoring and estimate the number of tags accurately without the requirement for any prior knowledge on the tag arrival and departure patterns. As another distinguished feature, the efficiency and stability of our framework in the sense of mean square is mathematically established.


\section{Technical Preliminaries}
\label{sec:pre}

In this section, we briefly introduce the extended Kalman filter and some fundamental concepts and results in stochastic process which are useful in the subsequent analysis. The main notations used in the paper are listed in Table~\ref{tab:notation}.

\begin{table}[!htbp]
\centering
\caption{Main Notations}\label{tab:notation}
\begin{tabular}{|l||l|l|l|}
\hline
$z_{k}$& System state in frame $k$: tag population \\
\hline
$y_{k}$& Measurement in frame $k$: idle slot frequency\\
\hline
$\hat z_{k+1|k}$& Priori prediction of $z_{k+1}$\\
\hline
$\hat z_{k|k}$& Posteriori estimate of $z_{k}$\\
\hline
$P_{k+1|k}$& Priori pseudo estimate covariance\\
\hline
$P_{k|k}$& Posteriori pseudo estimate covariance\\
\hline
$v_{k}$& Measurement residual in frame $k$\\
\hline
$K_{k}$& Kalman gain in frame $k$\\
\hline
$Q_{k}$, $R_{k}$& Two tunable parameters in frame $k$\\
\hline
$e_{k|k-1}$& Estimation error in frame ${k}$\\
\hline

\multicolumn{2}{|l|}{2. Defined in Sec.\ref{sec:model and formulation}, \ref{sec:static RFID} and \ref{sec:dynamic RFID}}\\
\hline
$L_{k}$& The length of frame $k$\\
\hline
$Rs_{k}$& Random seed in frame ${k}$\\
\hline
$h(\cdot)$& Hash function\\
\hline
$N_{k}$& The number of idle slots in frame ${k}$\\
\hline
$\rho$& Reader load factor\\
\hline
$p(z_k)$& Probability of an idle slot in frame ${k}$\\
\hline
$u_{k}$& Gaussian random variable\\
\hline
$Var[u_{k}]$& Variance of $u_{k}$\\
\hline
$\phi_k$& Controllable parameter\\
\hline
$w_k$& Random variable: variation of tag population \\
\hline
$\Phi_k$& Normalization of $v_{k}$\\
\hline
$g^{+}_k$, $g^{-}_k$& CUSUM statistics\\
\hline
$\theta$, $\Upsilon_k$& CUSUM threshold and reference value\\
\hline
\multicolumn{2}{|l|}{3. Defined in Sec.\ref{sec:analysis}}\\
\hline
$\epsilon$& Upper bound of initial estimation error\\
\hline
$\lambda_k$, $\delta_k$& Upper bounds of $E[w_k]$ and $E[w^2_k]$\\
\hline
\end{tabular}
\end{table}

\subsection{Extended Kalman Filter}

The extended Kalman filter is a powerful tool to estimate system state in nonlinear discrete-time systems. Formally, a nonlinear discrete-time system can be described as follows:
\begin{align}
z_{k+1} &= f(z_k, x_k) + w^*_k \\
y_k &= h(z_k) + u^*_k,
\end{align}
where $z_{k+1} \in \mathbb{R}^n$ denotes the state of the system, $x_k \in \mathbb{R}^d$ is the controlled inputs and $y_k \in \mathbb{R}^m$ stands for the measurement observed from the system. The uncorrelated stochastic variables $w^*_k \in \mathbb{R}^s$ and $u^*_k \in \mathbb{R}^t$ denote the process noise and the measurement noise, respectively. The functions $f$ and $h$ are assumed to be the continuously differentiable.

For the above system, we introduce an EKF-based state estimator given in Definition~\ref{Def:EKF}.

\begin{definition}[Extended Kalman filter~\cite{song1992EKF}]
A two-step discrete-time extended Kalman filter consists of state prediction and measurement update, defined as

1) Time update (prediction)
\begin{align}
\hat{z}_{k+1 | k} &= f(\hat{z}_{k | k},x_k) \label{Eq:Tz_up} \\
P_{k+1 | k} &= P_{k | k} + Q_k , \label{Eq:Tp_up}
\end{align}

2) Measurement update (correction)
\begin{align}
\hat{z}_{k+1 | k+1} &= f(\hat{z}_{k+1 | k},x_k) + K_{k+1} v_{k+1} \label{Eq:Mz_up}\\
P_{k+1 | k+1} &= P_{k+1 | k} \left( 1- K_{k+1} C_{k+1} \right) \label{Eq:Mp_up}\\
K_{k+1} &= \frac{P_{k+1 | k} C_{k+1}}{P_{k+1 | k} {C_{k+1}}^2 + R_{k+1} }, \label{Eq:gain}
\end{align}
where
\begin{align}
v_{k+1} &= y_{k+1} - h(\hat{z}_{k+1|k}) \label{Eq:innovation}\\
C_{k+1} &= \frac{\partial h({z}_{k+1})}{\partial {z}_{k+1}} \bigg |_{{z}_{k+1} = \hat{z}_{k+1|k} }.
\label{Eq:dh}
\end{align}
\label{Def:EKF}
\end{definition}

\begin{remark}
In the above definition of extended Kalman filter, the parameters can be interpreted in our context as follows:
\begin{itemize}
\item $\hat{z}_{k+1 | k}$ is the prediction of ${z}_{k+1}$ at the beginning of frame $k+1$ given by the previous state estimate, while $\hat{z}_{k+1 | k+1}$ is the estimate of ${z}_{k+1}$ after the adjustment based on the measure at the end of frame $k+1$.
\item $v_{k+1}$, referred to as innovation, is the measurement residual in frame $k$$+$$1$. It represents the estimated error of the measure.
\item $K_{k+1}$ is the Kalman gain. With reference to equation~\eqref{Eq:Mz_up}, it weighs the innovation $v_{k+1}$ w.r.t. $f(\hat{z}_{k+1 | k},x_k)$.
\item $P_{k+1|k}$ and $P_{k+1 | k+1}$, in contrast to the linear case, are not equal to the covariance of estimation error of the system state. In this paper, we will refer to them as pseudo-covariance.
\item $Q_k$ and $R_k$ are two tunable parameters which play the role as that of the covariance of the process and measurement noises in linear stochastic systems to achieve optimal filtering in the maximum likelihood sense.
    We will show later that $Q_k$ and $R_k$ also play an important role in improving the stability and convergence of our EKF-based estimators.
\end{itemize}
\end{remark}

\subsection{Boundedness of Stochastic Process}

In order to analyse the stability of an estimation algorithm, we need to check the boundedness of the estimation error defined as follows:
\begin{equation}
e_{k|k-1} \triangleq z_{k} - \hat{z}_{k|k-1}.
\label{Eq:error}
\end{equation}

We further introduce the following two mathematical definitions~\cite{morozan1972boundedness}~\cite{tarn1976observers} on the boundedness of stochastic process followed by stochastic stability results~\cite{reif1999stochastic}.

\begin{definition}[Boundedness of Random Variable]
The stochastic process $e_{k|k-1}$ is said to be bounded with probability one (w.p.o.), if there exists $X>0$ such that
\begin{equation}
\lim_{k \to \infty} \sup_{k \ge 1} \mathbb P\{|e_{k|k-1}|>X \} =0.
\end{equation}
\label{Def:boundedness}
\end{definition}

\begin{definition}[Boundedness in Mean Square]
The stochastic process $e_{k|k-1}$ is said to be exponentially bounded in the mean square with exponent $\zeta$, if there exist real numbers $\psi_1$, $\psi_2>0$ and $0<\zeta<1$ such that
\begin{equation}
E[e_{k|k-1}^2] \le \psi_1 e_{1|0}^2 {\zeta}^{k-1} + \psi_2.
\end{equation}
\label{Def:stability}
\end{definition}

\begin{lemma}
Given a stochastic process $V_k (e_{k|k-1})$ and real numbers $\uline \beta$, $\overline\beta$, $\tau$$>$$0$ and $0$$<$$\alpha$$\le$$1$ with the following properties:
\begin{equation}
  \uline \beta e_{k|k-1}^2 \le V_k (e_{k|k-1}) \le \overline\beta e_{k|k-1}^2 ,
\end{equation}
\begin{equation}
  E[V_{k+1} (e_{k+1|k}) | e_{k|k-1}] - V_k (e_{k|k-1}) \le -\alpha V_k (e_{k|k-1}) + \tau ,
\end{equation}
then for any $k \ge 1$ it holds that
\begin{itemize}
\item the stochastic process $e_{k|k-1}$ is exponentially bounded in the mean square, i.e.,
\begin{eqnarray}
E[e_{k|k-1}^2]
&\le& \frac{\overline\beta}{\uline \beta} E[e_{1|0}^2] (1-\alpha)^{k-1}
                     + \frac{\tau}{\uline \beta} \sum_{j=1}^{k-2} {(1-\alpha)^j} \nonumber\\
&\le& \frac{\overline\beta}{\uline \beta} E[e_{1|0}^2] (1-\alpha)^{k-1}
                     + \frac{\tau}{\uline \beta \alpha},
\label{Eq:stb_ineq}
\end{eqnarray}
\item the stochastic process $e_{k|k-1}$ is bounded w.p.o..
\end{itemize}
\label{Lem:stability}
\end{lemma}

It can be noted that Lemma~\ref{Lem:stability} can only be implemented offline. To address this limit, we adjust Lemma~\ref{Lem:stability} to an online version with time-varying parameters, which can be proven by the same method as in~\cite{tarn1976observers,rhudy2013online}.

\begin{lemma}
If there exist a stochastic process $V_k (e_{k|k-1})$ and parameters $\beta^*$, $\beta_k$, $\tau_k$$>$$0$ and $0$$<$$\alpha^*_k$$\le$$1$ with the following properties:
\begin{equation}
V_1(e_{1|0}) \le \beta^* {e}_{1|0}^2,
\label{Eq:bou_v0}
\end{equation}
\begin{equation}
\beta_k {e}_{k|k-1}^2 \le V_k (e_{k|k-1}),
\label{Eq:bou_vk}
\end{equation}
\begin{equation}
E[V_{k+1} (e_{k+1|k}) | e_{k|k-1}] - V_k (e_{k|k-1}) \le -\alpha^*_k V_k (e_{k|k-1}) + \tau_k;
\label{Eq:bou_Ev}
\end{equation}
then for any $k \ge 1$ it holds that
\begin{itemize}
\item the stochastic process $e_{k|k-1}$ is exponentially bounded in the mean square, i.e.,
\begin{multline}
E[e_{k|k-1}^2] \le \frac{\beta^*}{\beta_k} E[{e_{1|0}}^2]\prod_{i=1}^{k-1} {(1-\alpha^*_{i})} \\
                     + \frac{1}{\beta_k} \sum_{i=1}^{k-2} \tau_{k-i-1}\prod_{j=1}^{i}{(1-\alpha^*_{k-j})},
\end{multline}
\item the stochastic process $e_{k|k-1}$ is bounded w.p.o..
\end{itemize}
\label{Lem:Mstability}
\end{lemma}

\begin{remark}
The conditions in Lemma~\ref{Lem:Mstability} can be interpreted from the following two angles:
\begin{enumerate}
\item the bounds of $V_k(e_{k|k-1})$, i.e,~\eqref{Eq:bou_v0},~\eqref{Eq:bou_vk},
\item the bounds of the drift of $V_k(e_{k|k-1})$, i.e,~\eqref{Eq:bou_Ev}.
\end{enumerate}
\end{remark}

\section{System Model and Problem Formulation}
\label{sec:model and formulation}


\subsection{System Model}

Consider a RFID system consisting of a reader and a mass of tags operating on one frequency channel. The number of tags is unknown \textit{a priori} and can be constant or dynamic (time-varying), which we refer to as \textit{static} and \textit{dynamic} systems, respectively throughout the paper. The MAC protocol for the RFID system is the standard framed-slotted ALOHA protocol, where the standard \textit{Listen-before-Talk} mechanism is employed by the tags to respond the reader's interrogation~\cite{Finkenzelle2000}.

The reader initiates a series of frames indexed by an integer $k \in \mathbb{Z_+}$. Each individual frame, referred to as a round, consists of a number of slots. The reader starts frame $k$ by broadcasting a \texttt{begin-round} command with frame size $L_k$ and a random seed $Rs_k$. The frame size is the number of available slots for tags to choose in a round. We adopt a dynamic framed-slotted ALOHA protocol where the frame size $L_k$ is set to the estimated number of tags at the start of $k$th frame based on a tag population estimation scheme, as will be described later. It is well-known that such setting maximises the protocol efficiency.
When a tag receives a \texttt{begin-round} command, it uses a hash function $h(\cdot)$, $L_k$, $Rs_k$, and its ID to generate a random number $i$ in the range $[0, L_k -1]$ and reply in slot $i$ of frame $k$.\footnote{The outputs of the hash function have a uniform distribution such that the tag can choose any slot within the round with the equal probability.}

Since every tag picks its own response slot individually, there may be zero, one, or more than one tags transmitting in a slot, which are referred to as \textit{idle}, \textit{singleton}, and \textit{collision} slots, respectively. The reader is not assumed to be able to distinguish between a singleton or a collision slot, but it can detect an idle slot. We term both singleton and collision slots as \textit{occupied} slots throughout the paper. By collecting all replies in a frame, the reader can generate a bit-string $B_k$ illustrated as $B_k=\{\cdots | 0 | 0 | 1 | 0 | 1 | 1 | \cdots \}$, where `0' indicates an idle slot, and `1' stands for an occupied one.


Subsequently, the reader finalizes the current frame by sending an \texttt{end round} command. Based on the number of idle slots, i.e., the number of `0' in $B_k$, the reader runs the estimation algorithm, detailed in the next section, to trace the tag population and updates the frame size for the next frame $k+1$ before starting the next round.

\subsection{Tag Population Estimation Problem}

Our objective is to design a stable and accurate tag population estimation algorithm. By stable and accurate we mean that
\begin{itemize}
\item the estimation error of our algorithm is bounded in mean square in the sense of Definition~\ref{Def:stability} and the relative estimation error tends to zero;
\item the estimated population size converges to the real value with exponential rate.
\end{itemize}

Mathematically, we consider a large-scale RFID system of a reader and a set of tags with the unknown size $z_k$ in frame $k$ which can be static or dynamic. Denote by $\hat{z}_{k|k-1}$ the prior estimate of $z_k$ in the beginning of frame $k$, the reader sets $L_k = \hat{z}_{k|k-1}$ for frame $k$. At the end of frame $k$, the reader updates the estimate $\hat{z}_{k|k-1}$ to $\hat{z}_{k|k}$ by running the estimation algorithm. Our designed estimation scheme need to guarantee the following properties:
\begin{itemize}
\item $\displaystyle \lim_{z_k\rightarrow\infty} \left|\frac{\hat{z}_{k|k-1}-z_k}{z_k}\right|=0$;
\item the converges rate is exponential.
\end{itemize}

\section{Tag Population Estimation: Static Systems}
\label{sec:static RFID}

In this section, we focus on the baseline scenario of static systems where the tag population is constant during the estimation process. We first establish the discrete-time model for the system dynamics and the measurement model using the bit string $B_k$ observed during frame $k$. We then present our EKF-based estimation algorithm.


\subsection{System Dynamics and Measurement Model}


Consider the static RFID systems where the tag population stays constant, the system state evolves as
\begin{equation}
z_{k+1} = z_{k},
\label{Eq:in_system_state}
\end{equation}
meaning that the number of tags $z_{k+1}$ in the system in frame $k+1$ equals that in frame $k$.


In order to estimate $z_k$, we leverage the measurement on the number of idle slots during a frame.
To start, we study the stochastic characteristics of the number of idle slots.

Assume that the initial tag population $z_0$ falls in the interval $z_0 \in [ \uline {z}_0, \overline {z}_0]$, yet the exact value of $z_0$ is unknown and should be estimated. The range $[\uline {z}_0, \overline {z}_0]$ can be a very coarse estimation that can be obtained by any existing population estimation method.
Recall the system model that in frame $k$, the reader probes the tags with the frame size $L_k$. Denote by variable $N_k$ the number of idle slots in frame $k$, that is, the number of `0's in $B_k$, we have the following results on $N_k$ according to~\cite{Kolchin1978allocation,kodialam2006mobicom}.

\begin{lemma}
\label{Lem:Normal}
If each tag replies in a random slot among the $L_k$ slots, then it holds that $N_k \sim {\cal N}[ \mu, \sigma^2]$ for large $L_k$ and $z_k$, where
$\mu = L_k (1- \frac{1}{L_k})^{z_k}$
and
$\sigma^2 = L_k (L_k -1)(1- \frac{2}{L_k})^{z_k} + L_k (1- \frac{1}{L_k})^{z_k} - {L_k}^2 (1- \frac{1}{L_k})^{2 z_k}$.
\end{lemma}

\begin{lemma}
\label{Lem:app_exp}
For any $\epsilon^* >0$, there exists some $M>0$, such that if $z_k \ge M$ or $L_k =\hat{z}_{k|k-1} \ge M$, then it holds that
\begin{align}
\big|\mu - L_k e^{- \rho} \big| \le \epsilon^* , \label{Eq:app_1}\\
\big|\sigma^2 - L_k (e^{- \rho} - (1+ \rho)e^{- 2\rho}) \big| \le \epsilon^*, \label{Eq:app_2}
\end{align}
where $\rho = \frac{z_k}{L_k}$ is referred to as the reader load factor.
\end{lemma}

Lemmas~\ref{Lem:Normal} and~\ref{Lem:app_exp} imply that in large-scale RFID systems, we can use $L_k e^{- \rho}$ and $L_k (e^{- \rho} - (1+ \rho)e^{- 2\rho})$ to approximate $\mu $ and $\sigma^2$.

At the end of each frame $k$, the reader gets a measure $y_k$ of the idle slot frequency defined as
\begin{equation}
y_k = \frac{N_k}{L_k}.
\label{Eq:y_k}
\end{equation}
Recall Lemma~\ref{Lem:Normal}, it holds that $y_k$ is a Normal distributed random variable specified as follows:
\begin{align*}
E[y_k] &=  e^{- \rho},\\
Var[y_k] &= \frac{1}{{L_k}} (e^{- \rho} - (1+ \rho)e^{- 2\rho}).
\end{align*}

Since there are $z_k$ tags that reply in frame $k$, the probability that a slot is idle, denoted as $p(z_k)$, can be calculated as
\begin{equation}
p(z_k) = (1- \frac{1}{L_k})^{z_k} \approx e^{-\frac{z_k}{L_k}}.
\label{Eq:h-in}
\end{equation}
Notice that for large $z_k$, $p(z_k)$ can be regarded as a continuously differentiable function of $z_k$.

Using the language in the Kalman filter, we can write $y_k$ as follows:
\begin{equation}
y_k = p(z_k) + u_k,
\label{Eq:measurement}
\end{equation}
where, based on the statistic characteristics of $y_k$, $u_k$ is a Gaussian random variable with zero mean and variance
\begin{equation}
Var[u_k] = \frac{1}{{L_k}} (e^{- \rho} - (1+ \rho)e^{- 2\rho}).
\end{equation}
We note that $u_k$ measures the uncertainty of $y_k$.

To summarise, the discrete-time model for static RFID systems is characterized by~\eqref{Eq:in_system_state} and~\eqref{Eq:measurement}.
We conclude this subsection by stating the following auxiliary lemma which is useful in our later analysis.

\begin{lemma}
Denote the function
$$\Lambda(\rho) \triangleq Var[u_k]=e^{- \rho} - (1+ \rho)e^{- 2\rho}, \rho>0,$$
it holds that $\Lambda(\rho)$ has a unique maximiser $\rho^* \in (1,2)$, i.e.,
\begin{equation}
Var[u_k] \le \frac{1}{{L_k}} \Lambda(\rho^*).
\label{Eq:Upp_u}
\end{equation}
Furthermore, it holds that $\Lambda(\rho^*) < e^{-2}$.
\label{Lem:Upp_u}
\end{lemma}

\begin{proof}
We compute the derivative of $\Lambda(\rho)$:
\begin{equation}
\frac{\mathrm{d}\Lambda}{\mathrm{d}\rho}
= e^{-2 \rho} \left( 2 \rho +1 -e^{\rho} \right).
\end{equation}
Noticing that
\begin{align}
\frac{\mathrm{d}(2 \rho +1 -e^{\rho})}{\mathrm{d}\rho}
= 2 - e^{\rho}
\begin{cases}
>0, \rho < \ln 2,\\
=0, \rho = \ln 2,\\
<0, \rho > \ln 2,\\
\end{cases}
\end{align}
and that $2 \rho +1 -e^{\rho} >0$ for $\rho \le 1$, $2 \rho +1 -e^{\rho} < 0$ for $\rho \ge 2$, we have
\begin{itemize}
\item $\frac{d\Lambda}{d \rho} >0$, for $0<\rho \le1$ and $\frac{d\Lambda}{d \rho} <0$, for $\rho\ge 2$;
\item $\frac{d\Lambda}{d \rho}$ is monotonously decreasing for $\rho>\ln 2$,
\end{itemize}
Hence, there is a unique solution $\rho^* \in (1,2)$ for $\frac{\mathrm{d}\Lambda}{\mathrm{d}\rho}=0$, i.e. $\Lambda(\rho^*) \ge \Lambda(\rho), \forall\rho >0$.

We now prove $\Lambda(\rho^*) < e^{-2}$. For $\rho* \in (1,2)$, since
\begin{align*}
\frac{\mathrm{d}\left(\rho^* e^{-\rho^*}+e^{-\rho^*}+e^{\rho^*-2}\right)}{\mathrm{d}\rho^*}
&=   e^{-\rho^*} \left( e^{(2\rho^*-2)} -\rho^*\right) \\
&\ge e^{-\rho^*} ( \rho^*-1) > 0,
\end{align*}
it holds that $\rho^* e^{-\rho^*}+e^{-\rho^*}+e^{\rho^*-2} > \frac{3}{e}$ for $\rho* \in (1,2)$.
We thus get
\begin{align}
\Lambda(\rho^*) - e^{-2}
&=  e^{-\rho^*} \left(1-\rho^* e^{-\rho^*}-e^{-\rho^*}-e^{\rho^*-2}\right) \nonumber\\
&<  e^{-\rho^*} \left(1-\frac{3}{e}\right) <  0,
\end{align}
which leads to $\Lambda(\rho^*) < e^{-2}$.
\end{proof}

\subsection{Tag Population Estimation Algorithm}

Noticing that the system state characterised by~\eqref{Eq:in_system_state} and~\eqref{Eq:measurement} is a discrete-time nonlinear system, we thus leverage the two-step EKF described in Definition~\ref{Def:EKF} to estimate the system state.
In~\eqref{Eq:gain}, the Kalman gain $K_k$ increases with $Q_k$ while decreases with $R_k$. As a result, $Q_k$ and $R_k$ can be used to tune the EKF such that increasing $Q_k$ and/or decreasing $R_k$ accelerates the convergence rate but leads to larger estimation error. In our design, we set $Q_k$ to a constant $q>0$ and introduce a parameter $\phi_k$ as follows to replace $R_k$ to facilitate our demonstration:
\begin{equation}
R_k = \phi_k P_{k|k-1} {C_k}^2.
\label{Eq:Rk}
\end{equation}
It can be noted from~\eqref{Eq:gain} and~\eqref{Eq:Rk} that $K_k$ is monotonously decreasing in $\phi_k$, i.e., a small $\phi_k$ leads to quick convergence with the price of relatively high estimation error. Hence, choosing the appropriate value for $\phi_k$ consists of striking a balance between the convergence rate and the estimation error. In our work, we take a dynamic approach by setting $\phi_k$ to a small value $\uline \phi$ but satisfying~\eqref{Eq:phi_var} at the first few rounds ($J$ rounds) of estimation to allow the system to act quickly since the estimation in the beginning phase can be very coarse. After that we set $\phi_k$ to a relatively high value $\overline \phi$ to achieve high estimation accuracy.

\begin{algorithm}[H]
\caption{Tag population estimation (static cases): executed by the reader}
\begin{algorithmic}[1]
\REQUIRE{$\uline {z}_0$, $P_{0|0}$, $q$, $J$, $\uline \phi$, $\overline \phi$, maximum number of rounds $k_{max}$}
\ENSURE{Estimated tag population set $S_z =\{\hat{z}_{k|k}: k \in [0, k_{max}] \}$ }
\STATE \textbf{Initialisation:} $\hat{z}_{0|0} \leftarrow \uline {z}_0$, $Q_0\leftarrow q$, $S_z = \{\hat{z}_{0|0}\}$
\FOR{$k=1$ to $k_{max}$}
     \STATE $\hat{z}_{k|k-1} \leftarrow \hat{z}_{k-1|k-1}$, $L_{k} \leftarrow \hat{z}_{k|k-1}$,
           \\$P_{k|k-1} \leftarrow P_{k-1|k-1} + Q_{k-1}$
     \STATE Generate a new random seed $Rs_{k}$
     \STATE Broadcast ($L_{k}$, $Rs_k$)
     \STATE Run \textit{Listen-before-Talk} protocol
     \STATE Obtain the number of idle slots $N_{k}$, \\
            Compute $y_{k}$ and $v_{k}$ using~\eqref{Eq:y_k} and~\eqref{Eq:innovation}
     \STATE $Q_k \leftarrow q$
     \IF{$k \le J$}
        \STATE $\phi_k \leftarrow \uline \phi$
     \ELSE
        \STATE $\phi_k \leftarrow \overline \phi$
     \ENDIF
     \STATE Calculate $R_{k}$ and $K_{k}$ using~\eqref{Eq:Rk} and~\eqref{Eq:gain}
     \STATE Update $\hat{z}_{k|k}$ and $P_{k|k}$ using~\eqref{Eq:Mz_up} and~\eqref{Eq:Mp_up}
     \STATE $S_z \leftarrow S_z \ \mathbb{\cup} \ \{\hat{z}_{k|k}\}$
\ENDFOR
\end{algorithmic}
\label{Al:SEKF-in}
\end{algorithm}
\vspace{-0.2cm}
Now, we are ready to present our tag population estimation algorithm as illustrated in Algorithm~\ref{Al:SEKF-in}. The major procedures of our estimation algorithm can be summarised as:
\begin{enumerate}
\item \textit{In the beginning of frame $k$: prediction.} The reader first predicts the tag population based on the estimation at the end of frame $k$$-$$1$. It sets the frame length $L_k$ to the predicted value.
\item The reader launches the \textit{Listen-before-talk} protocol for frame $k$.
\item \textit{At the end of frame $k$: correction}. The reader computes $N_k$ based on $B_k$ and further calculates $y_k$ and $v_k$ from $N_k$. It then updates the prediction with the corrected estimate $\hat{z}_{k|k}$ following~\eqref{Eq:Mz_up}.
\end{enumerate}

We will theoretically establish the stability and accuracy of the estimation algorithm in Sec.~\ref{sec:analysis}.

\section{Tag Population Estimation: Dynamic Systems}
\label{sec:dynamic RFID}

In this section, we further tackle the dynamic case where the tag population may vary during the estimation process. We first establish the system model and then present our estimation algorithm.

\subsection{System Dynamics and Measurement Model}

In dynamic RFID systems, we can formulate the system dynamics as
\begin{equation}
z_{k+1} = z_{k} + w_k,
\label{Eq:system_state}
\end{equation}
where the tag population $z_{k+1}$ in frame $k$$+$$1$ consists of two parts: i) the tag population in frame $k$ and ii) a random variable $w_k$ which accounts for the stochastic variation of tag population resulting from the tag arrival/departure during frame $k$. Notice that $w_k$ is referred to as process noise in Kalman filters and the appropriate characterisation of $w_k$ is crucial in the design of stable Kalman filters, which will be investigated in detail later. Besides, the measurement model is the same as the static case. Hence, the discrete-time model for dynamic RFID systems can be characterized by~\eqref{Eq:system_state} and~\eqref{Eq:measurement}.

\subsection{Tag Population Estimation Algorithm}

In the dynamic case, we leverage the two-step EKF to estimate the system state combined with the CUSUM test to further trace the tag population fluctuation.

Our main estimation algorithm is illustrated in Algorithm~\ref{Al:SEKF}. The difference compared to the static scenario is that tag population variation needs to be detected by the CUSUM test presented in Algorithm~\ref{Al:CUSUM} in the next subsection and the output of Algorithm~\ref{Al:CUSUM} acts as a feedback to $\phi_k$ because due to the tag population variation, $\phi_k$ is no more a constant after the $J$th round as the static case. The overall structure of the estimation algorithm is illustrated in Fig.~\ref{fig:EKF}. We note that in the case where $z_k$ is constant, Algorithm~\ref{Al:SEKF} degenerates to Algorithm~\ref{Al:SEKF-in}.


\begin{algorithm}[H]
\caption{Tag population estimation (unified framework): executed by the reader}
\begin{algorithmic}[1]
\REQUIRE{$\uline {z}_0$, $P_{0|0}$, $q$, $J$, $\uline \phi$, $\overline \phi$, maximum number of rounds $k_{max}$}
\ENSURE{Estimation set $S_z =\{\hat{z}_{k|k}: k \in [0, k_{max}] \}$ }
\STATE \textbf{Initialisation:} $\hat{z}_{0|0} \leftarrow \uline {z}_0$, $Q_0\leftarrow q$, $S_z = \{\hat{z}_{0|0}\}$
\FOR{$k=1$ to $k_{max}$}
     \STATE $\hat{z}_{k|k-1} \leftarrow \hat{z}_{k-1|k-1}$, $L_{k} \leftarrow \hat{z}_{k|k-1}$,
           \\$P_{k|k-1} \leftarrow P_{k-1|k-1} + Q_{k-1}$
     \STATE Generate a new random seed $Rs_{k}$
     \STATE Broadcast ($L_{k}$, $Rs_k$)
     \STATE Run \textit{Listen-before-Talk} protocol
     \STATE Obtain the number of idle slots $N_{k}$, \\
            Compute $y_{k}$ and $v_{k}$ using~\eqref{Eq:y_k} and~\eqref{Eq:innovation}
     \STATE $Q_k \leftarrow q$
     \IF{$k \le J$}
        \STATE $\phi_k \leftarrow \uline \phi$
     \ELSE
        \STATE Execute Algorithm~\ref{Al:CUSUM}
        \STATE $\phi_{k} \leftarrow \textit{output of Algorithm~\ref{Al:CUSUM}}$
     \ENDIF
     \STATE Calculate $R_{k}$ and $K_{k}$ using~\eqref{Eq:Rk} and~\eqref{Eq:gain}
     \STATE Update $\hat{z}_{k|k}$ and $P_{k|k}$ using~\eqref{Eq:Mz_up} and~\eqref{Eq:Mp_up}
     \STATE $S_z \leftarrow S_z \ \mathbb{\cup} \ \{\hat{z}_{k|k}\}$
\ENDFOR
\end{algorithmic}
\label{Al:SEKF}
\end{algorithm}

\begin{figure}[htbp]
\centering
\includegraphics[width=6cm]{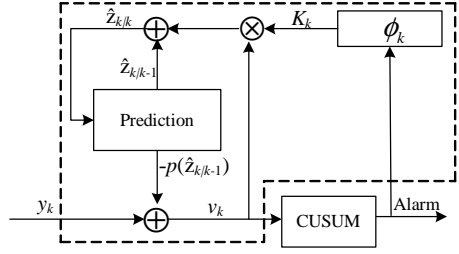}
\caption{Estimation algorithm diagram: Dashed box indicates the EKF.}
\label{fig:EKF}
\end{figure}

\subsection{Detecting Tag Population Change: CUSUM Test}

\noindent\textbf{The CUSUM Detection Framework.}

We leverage the CUSUM test to detect the change of tag population and further adjust $\phi_k$. CUSUM test is a sequential analysis technique typically used for change detection~\cite{gustafsson2000CUSUM}. It is shown to be asymptotically optimal in the sense of the minimum detection time subject to a fixed worst-case expected false alarm rate~\cite{brodsky1993nonparametric}.

In the context of dynamic tag population detection, the reader monitors the innovation process $v_k=y_k- p(\hat z_{k|k-1})$. If the number of the tags population is constant, $v_k$ equals to $u_k$ which is a Gaussian process with zero mean. In contrast, upon the system state changes, i.e., tag population changes, $v_k$ drifts away from the zero mean. In our design, we use $\Phi_k$ as a normalised input to the CUSUM test by normalising $v_k$ with its estimated standard variance, specified as follows:
\begin{equation}
\Phi_k = \frac{v_k}{\sqrt{(P_{k|k-1}+Q_{k-1}){C_k}^2 + Var[u_k] \big|_{z_k = \hat z_{k|k-1}}}}.
\label{Eq:CUSUM_inov}
\end{equation}
The reader further updates the CUSUM statistics ${g}_k^+$ and ${g}_k^-$ as follows:
\begin{align}
{g}_k^+ &= \max \{0, {g}_{k-1}^+ +\Phi_k - \Upsilon\}, \label{Eq:CUSUM_g+}\\
{g}_k^- &= \min \{0, {g}_{k-1}^- +\Phi_k +\Upsilon\},\  \label{Eq:CUSUM_g-}\\
{g}_k^+ &= {g}_k^- = 0, \text{if} \  \delta =1,
\end{align}
where ${g}_0^+$$=$$0$ and ${g}_0^-$$=0$. And $\Upsilon$$\ge$$0$, referred to as reference value, is a filter design parameter indicating the sensitivity of the CUSUM test to the fluctuation of $\Phi_k$, Moreover, by $\delta$ we define an indicator flag indicating tag population change:
\begin{eqnarray}
\delta=
\begin{cases}
1 & \text{if} \ {g}_k^+ > \theta \ \text{or}\ {g}_k^- <-\theta, \\
0 & \text{otherwise},
\end{cases}
\end{eqnarray}
where $\theta >0$ is a pre-specified CUSUM threshold.

The detailed procedure of the change detection is illustrated in Algorithm~\ref{Al:CUSUM}.
\begin{algorithm}[H]
\caption{CUSUM test: executed by the reader in frame $k$}
\begin{algorithmic}[1]
\REQUIRE{$\Upsilon$, $\theta$}
\ENSURE{$\phi_k$}
\STATE \textbf{Initialisation:} ${g}_0^+ \leftarrow 0$, ${g}_0^- \leftarrow 0$
    \STATE Compute $\Phi_k$ using equation~\eqref{Eq:CUSUM_inov}
    \STATE ${g}_k^+\leftarrow $~\eqref{Eq:CUSUM_g+}, ${g}_k^- \leftarrow$~\eqref{Eq:CUSUM_g-}
    \IF{${g}_k^+ >\theta$ or ${g}_k^- <-\theta$}
         \STATE    $\delta \leftarrow 1$, $\phi_k \leftarrow \varphi_1(\delta)$, ${g}_k^+ \leftarrow 0$, ${g}_{k}^-\leftarrow 0$
    \ELSE
         \STATE    $\delta \leftarrow 0$, $\phi_k \leftarrow \varphi_1(\delta)$
    \ENDIF
    \STATE Return $\phi_k$
\end{algorithmic}
\label{Al:CUSUM}
\end{algorithm}

\noindent\textbf{Parameter tuning in CUSUM test.}

The choice of the threshold $\theta$ and the drift parameter $\Upsilon$ has a directly impact on the performance of the CUSUM test in terms of detection delay and false alarm rate. Formally, the average running length (ARL) $L(\mu^*)$ is used to denote the duration between two actions~\cite{basseville1993abruptdetection}. For a large $\theta$, $L(\mu^*)$ can be approximated as
\footnote{For two variables X, Y, asymptotic notation $X=\Theta(Y)$ implies that there exist positives $c_1$, $c_2$ and $x_0$ such that for $\forall X>x_0$, it follows that $c_1 X \le Y \le c_2 X$.}
\begin{eqnarray}
L(\mu^*) =
\begin{cases}
\Theta(\theta),& \text{if}\ \mu^* \neq 0,\\
\Theta(\theta^2),& \text{if}\ \mu^* = 0,
\end{cases}
\label{Eq:ARL}
\end{eqnarray}
where $\mu^*$ denotes the mean of the process $\Phi_k$.

In our context, ARL corresponds to the mean time between two false alarms in the static case and the mean detection delay of the tag population change in the dynamic case. It is easy to see from~\eqref{Eq:ARL} that a higher value of $\theta$ leads to lower false alarm rate at the price of longer detection delay. Therefore, the choices of $\theta$ and $\Upsilon$ consists of a tradeoff between the false alarm rate and the detection delay.

Recall that $\Phi_k$ can be approximated to a white noise process, i.e, $\Phi_k \sim {\cal N}[ \mu^*, {\sigma^*}^2]$ with $\mu^*=0$, $\sigma^*=1$ if the system state does not change. Generically, as recommended in~\cite{spiring2007quality}, setting $\theta$ and $\Upsilon$ as follows achieves good ARL from the engineering perspective.
\begin{eqnarray}
\theta &=& 4 \sigma^*,
\label{Eq:theta}\\
\Upsilon &=& \mu^* + 0.5 \sigma^*.
\label{Eq:Upsilon}
\end{eqnarray}

In the CUSUM framework, we set $\phi_k$ by $\varphi_1(\delta)$ as follows:
\begin{eqnarray}
\varphi_1(\delta) =
\begin{cases}
\uline \phi,& \text{if}\ \delta =1, \\
\overline \phi,& \text{if}\ \delta =0.
\end{cases}
\end{eqnarray}
The rationale is that once a change on the tag population is detected in frame $k$, $\phi_k$ is set to $\uline \phi$ to quickly react to the change, while $\phi_k$ sticks to $\overline \phi$ when no system change is detected.

\section{Performance Analysis}
\label{sec:analysis}

In this section, we establish the stability and the accuracy of our estimation algorithms for both static and dynamic cases.

\subsection{Static Case}

Our analysis is composed of two steps. First, we derive the estimation error. Second, we establish the stability and the accuracy of Algorithm~\ref{Al:SEKF} in terms of the boundedness of estimation error.

\noindent\textbf{Computing Estimation Error.}

We first approximate the non-linear discrete system by a linear one. To that end, as the function $p(z_{k})$ is continuously differentiable at $z_{k} = \hat{z}_{k|k-1}$, using the Taylor expansion and the fact that $L_{k} = \hat{z}_{k|k-1}$, we have
\begin{equation}
p(z_{k}) = p(\hat{z}_{k|k-1}) + C_{k} (z_{k}-\hat{z}_{k|k-1}) + \chi (z_{k},\hat{z}_{k|k-1}),
\label{Eq:linear_h}
\end{equation}
where
\begin{align}
&C_{k} = -\frac{1}{e \hat{z}_{k|k-1}}, \label{Eq:ck}\\
&\chi (z_{k},\hat{z}_{k|k-1}) = \sum_{j=2}^{\infty} { \frac{1}{e j!}(1-\frac{z_{k}}{\hat{z}_{k|k-1}})^j}.\label{Eq:resdual}
\end{align}
Regarding the convergence of $\chi (z_{k},\hat{z}_{k|k-1})$ in~\eqref{Eq:resdual}, assume that
\begin{equation}
z_{k} = a'_k \hat{z}_{k|k-1}, 0< a'_k < 2,
\label{Eq:a'}
\end{equation}
we can further obtain the boundedness of the residual for the case $z_{k} < 2\hat{z}_{k|k-1}$ by some algebraic operations as follows:
\begin{align}
|\chi (z_{k},\hat{z}_{k|k-1})|
& = \frac{(\hat{z}_{k|k-1} - z_{k})^2}{e {\hat{z}_{k|k-1}}^2}\sum_{j=0}^{\infty} { \frac{1}{
    (j+2)!}\Big|1-\frac{z_{k}}{\hat{z}_{k|k-1}} \Big|^j} \nonumber\\
&\le \frac{(\hat{z}_{k|k-1} - z_{k})^2}{2e { \hat{z}_{k|k-1}}^2} \sum_{j=0}^{\infty}
    {\Big|1-\frac{z_{k}}{\hat{z}_{k|k-1}} \Big|^j} \nonumber\\
&\le \frac{(\hat{z}_{k|k-1} - z_{k})^2}{2e { \hat{z}_{k|k-1}}^2} \frac{1}{1-\Big|1-\frac{z_{k}}{\hat{z}_{k|k-1}} \Big|} \nonumber\\
&\le \frac{(\hat{z}_{k|k-1} - z_{k})^2}{2e a_k { \hat{z}_{k|k-1}}^2},
\end{align}
where
\begin{eqnarray}
a_k =
\begin{cases}
a'_k & \text{if}\ 0< a'_k \le 1,\\
2-a'_k & \text{if}\ 1<a'_k<2.
\end{cases}
\label{Eq:aa'}
\end{eqnarray}

Recall the definition of the estimation error in~\eqref{Eq:error} and using~\eqref{Eq:in_system_state}, ~\eqref{Eq:Tz_up} and~\eqref{Eq:Mz_up}, we can derive the estimation error $e_{k+1|k}$ as follows:
\begin{align}
e_{k+1|k}
=& z_{k+1} - \hat z_{k+1|k}  = z_k - \hat z_{k|k}  = z_k - \hat z_{k|k-1} \nonumber\\
   &- K_k \left[ C_k (z_k - \hat z_{k|k-1}) + \chi (z_{k},\hat{z}_{k|k-1}) + u_k\right] \nonumber\\
=& (1- K_k C_k) e_{k|k-1} + s_k + m_k,
\label{Eq:error_k+1}
\end{align}
where $s_k$ and $m_k$ are defined as
\begin{align}
s_k &= - K_k u_k \label{Eq:sk},\\
m_k &= - K_k \chi (z_{k},\hat{z}_{k|k-1}).
\label{Eq:mk}
\end{align}

\noindent\textbf{Boundedness of Estimation Error}

Having derived the dynamics of the estimation error, we now state the main result on the stochastic stability and accuracy of Algorithm~\ref{Al:SEKF-in}.

\begin{theorem}
Consider the discrete-time stochastic system given by~\eqref{Eq:in_system_state} and~\eqref{Eq:measurement} and Algorithm~\ref{Al:SEKF-in}, the estimation error $e_{k|k-1}$ defined by~\eqref{Eq:error} is exponentially bounded in mean square and bounded w.p.o., if the following conditions hold:
\begin{enumerate}
\item there are positive numbers $\uline q$, $\overline q$, $\uline \phi$ and $\overline \phi$ such that the bounds on $Q_k$ and $\phi_k$ are satisfied for every $k$$\ge$$0$, as in
    \begin{align}
    \uline q    &\le  Q_k \le \overline q, \label{Eq:bound_Q}\\
    \uline \phi &\le  \phi_k \le \overline \phi, \label{Eq:bound_phi}
    \end{align}
\item The initialization must follow the rules
    \begin{align}
    P_{0|0} &>0 \label{Eq:P0}, \\
    |e_{1|0}| &\le \epsilon \label{Eq:e0}
    \end{align}
    with positive real number $\epsilon>0$.
\end{enumerate}
\label{Th:conditions}
\end{theorem}

\begin{remark}
By referring to the design objective posed in Section~\ref{sec:model and formulation}, Theorem~\ref{Th:conditions} prove the following properties of our estimation algorithm:
\begin{itemize}
\item the estimation error of our algorithm is bounded in mean square and the relative estimation error tends to zero;
\item the estimated population size converges to the real value with exponential rate.
\end{itemize}

The conditions in Theorem~\ref{Th:conditions} can be interpreted as follows:
\begin{enumerate}
\item The inequalities~\eqref{Eq:bound_Q} and~\eqref{Eq:bound_phi} can be satisfied by the configuring the correspondent parameters in Algorithm~\ref{Al:SEKF-in}, which guarantees the boundedness of the pseudo-covariance $P_{k|k-1}$ as shown later.
\item  The inequality~\eqref{Eq:P0} consists of establishing positive $P_{k|k-1}$ for every $k \ge 1$.
\item As a sufficient condition for stability, the upper bound $\epsilon$ may be too stringent. As shown in the simulation results, stability is still ensured even with a relatively large $\epsilon$.
\end{enumerate}
\end{remark}

Before the proof of Theorem~\ref{Th:conditions}, we first state several auxiliary lemmas to streamline the proof.

\begin{lemma}
Under the conditions of Theorem~\ref{Th:conditions}, if $P_{0|0}>0$, there exist $\uline p_{k},\overline p_{k}>0$ such that the pseudo-covariance $P_{k|k-1}$ is bounded for every $k \ge 1$, i.e.,
\begin{equation}
\uline p_{k} \le  P_{k|k-1} \le \overline p_k.
\label{Eq:bound_P}
\end{equation}
\label{Lem:bound_P}
\end{lemma}

\begin{IEEEproof}
Recall~\eqref{Eq:Tp_up} and~\eqref{Eq:Mp_up}, we have
\begin{equation}
P_{k|k-1} \ge Q_{k-1},
\end{equation}
and
\begin{align}
P_{k|k-1}
&= P_{k-1|k-1} + Q_{k-1} \nonumber \\
&= P_{k-1|k-2} (1-K_{k-1} C_{k-1}) + Q_{k-1} \nonumber \\
&= P_{k-1|k-2} \left( 1- \frac{P_{k-1|k-2} {C_{k-1}}^2}{{P_{k-1|k-2} {C_{k-1}}^2 + R_{k-1}}}\right) + Q_{k-1}.
\label{Eq:P_Q}
\end{align}
Following the design of $R_k$ in~\eqref{Eq:Rk} and by iteration, we further get
\begin{align*}
&P_{k|k-1}
= P_{k-1|k-2} \left( 1-\frac{1}{1+ \phi_{k-1}} \right) + Q_{k-1} \nonumber\\
& = P_{1|0} \prod_{i=1}^{k-1} \left( 1-\frac{1}{1+ \phi_{i}} \right) + \sum_{i=0}^{k-2} {Q_i \prod_{j=i}^{k-2} \left( 1-\frac{1}{1+ \phi_{j+1}} \right)}
   + Q_{k-1}.
\label{Eq:P_phi}
\end{align*}

Since $\phi_k$ and $Q_k$ are controllable parameters, we can set $\phi_k \le \overline \phi$ and $Q_k \le \overline q$ for every $k \ge 0$ in Algorithm~\ref{Al:SEKF-in}, where $\overline \phi,\overline q >0$. Consequently, we have
\begin{align}
P_{k|k-1}
& \le P_{1|0} \left( 1-\frac{1}{1+ \overline \phi} \right)^{k-1}
      + {\overline q}\sum_{j=1}^{k-1} {\left( 1-\frac{1}{1+ \overline \phi} \right)^j} + Q_{k-1} \nonumber\\
& \le (P_{0|0} + Q_0) \left( 1-\frac{1}{1+ \overline \phi} \right)^{k-1}
      + {\overline q} {\overline \phi} + Q_{k-1}
\end{align}

Let $\overline p_{k}= ((P_{0|0} + Q_0) \left( 1-\frac{1}{1+ \overline \phi} \right)^{k-1}+ {\overline q} {\overline \phi} + Q_{k-1}$ and $\uline p_{k} = Q_{k-1}$, we thus have $\uline p_{k} \le P_{k|k-1} \le \overline p_{k}$.
\end{IEEEproof}

\begin{lemma}
Let $\alpha_k\triangleq\frac{1}{1+ \phi_{k}}$, it holds that
\begin{equation}
\frac{(1-K_k C_k)^2}{P_{k+1|k}} e_{k|k-1}^2 \le (1-\alpha_k) \frac{e_{k|k-1}^2}{P_{k|k-1}},\ \forall k\ge 1.
\label{Eq:alpha}
\end{equation}
\label{Lem:alpha}
\end{lemma}

\begin{IEEEproof}
From~\eqref{Eq:P_Q}, we have
\begin{align}
P_{k+1|k}
&= P_{k|k-1} \left( 1-K_k C_k \right) + Q_{k} \nonumber \\
&\ge P_{k|k-1} \left( 1-K_k C_k \right).
\label{Eq:PKC}
\end{align}
By substituting it into the left-hand side of~\eqref{Eq:alpha} and using the fact that $R_k = \phi_{k} P_{k|k-1} {C_{k}}^2$ for every $k \ge1$, we get
\begin{align}
\frac{(1-K_k C_k)^2}{P_{k+1|k}} e_{k|k-1}^2
&\le \frac{(1-K_k C_k)^2}{P_{k|k-1} \left( 1-K_k C_k \right)} e_{k|k-1}^2 \nonumber\\
&\le (1-K_k C_k) \frac{{e_{k|k-1}}^2 }{P_{k|k-1}} \nonumber\\
&\le \left( 1-\frac{1}{1+ \phi_{k}} \right) \frac{e_{k|k-1}^2 }{P_{k|k-1}}.
\end{align}
We are thus able to prove~\eqref{Eq:alpha}.
\end{IEEEproof}

\begin{lemma}
Let $\displaystyle b_k \triangleq \frac{4 a_k \phi_{k}+|a'_k-1|}{4 a_k^2 \phi_{k}(1+\phi_{k}) \hat{z}_{k|k-1} P_{k|k-1}}$, it holds that
\begin{equation}
\frac{m_k [2(1-K_k C_k)e_{k|k-1}+m_k]}{P_{k+1|k}} \le b_k |\hat{z}_{k|k-1} - z_{k}|^3
\label{Eq:le_mk}
\end{equation}
for any $a_k'\in (0,2)$.
\label{Lem:mk}
\end{lemma}

\begin{IEEEproof}
From~\eqref{Eq:mk}, we get the following expansion
\begin{multline}
\frac{m_k [2(1-K_k C_k)e_{k|k-1}+m_k]}{P_{k+1|k}} \\
= \frac{1}{P_{k+1|k}} \frac{-P_{k|k-1} C_k}{P_{k|k-1} {C_k}^2 + R_k} \chi(z_k,\hat z_{k|k-1})\\
  \cdot \Bigg[ 2\left( 1-  \frac{P_{k|k-1} {C_k}^2}{P_{k|k-1} {C_k}^2 + R_k}\right) e_{k|k-1} \\
   - \frac{P_{k|k-1} C_k}{P_{k|k-1} {C_k}^2 + R_k} \chi(z_k,\hat z_{k|k-1}) \Bigg].
\end{multline}
It then follows from~\eqref{Eq:ck},~\eqref{Eq:a'} and~\eqref{Eq:PKC} that
\begin{align*}
&\frac{m_k [2(1-K_k C_k)e_{k|k-1}+m_k]}{P_{k+1|k}} \\
&\le \frac{1}{P_{k|k-1}( 1-K_k C_k)} \frac{-P_{k|k-1} C_k}{P_{k|k-1} {C_k}^2 + R_k}
          \frac{(\hat{z}_{k|k-1} - z_{k})^2}{2e a_k { \hat{z}_{k|k-1}}^2}\\
  &\cdot \Bigg[ 2\left( 1-  \frac{P_{k|k-1} {C_k}^2}{P_{k|k-1} {C_k}^2 + R_k}\right) |\hat{z}_{k|k-1} - z_{k}| \\
   &  - \frac{P_{k|k-1} C_k}{P_{k|k-1} {C_k}^2 + R_k}
             \frac{(\hat{z}_{k|k-1} - z_{k})^2}{2e a_k { \hat{z}_{k|k-1}}^2} \Bigg] \\
&\le \frac{1+\phi_{k}}{\phi_{k} P_{k|k-1}} \frac{-1}{C_k (1+\phi_{k})}
               \frac{(\hat{z}_{k|k-1} - z_{k})^2}{2e a_k { \hat{z}_{k|k-1}}^2} \\
     & \cdot \Bigg[ \frac{2\phi_{k}|\hat{z}_{k|k-1} - z_{k}|}{1+\phi_{k}}
     - \frac{1}{C_k (1+\phi_{k})}
               \frac{(\hat{z}_{k|k-1} - z_{k})^2}{2e a_k { \hat{z}_{k|k-1}}^2} \Bigg] \\
&\le \frac{1}{\phi_{k} P_{k|k}} \frac{|\hat{z}_{k|k-1} - z_{k}|^3}{2 a_k \hat{z}_{k|k-1}}
          \cdot  \frac{4 a_k \phi_{k}+|a'_k-1|}{2 a_k (1+\phi_{k})}.
\end{align*}
We are thus able to prove~\eqref{Eq:le_mk}.
\end{IEEEproof}

\begin{lemma}
$\displaystyle E \left[\frac{{s_k}^2}{P_{k+1|k}} \big| e_{k|k-1} \right]
\le \frac{\hat{z}_{k|k-1}}{\phi_{k} (1+\phi_{k}) P_{k|k-1}}$.
\label{Lem:E_s2}
\end{lemma}

\begin{proof}
From~\eqref{Eq:sk}, we have
\begin{equation}
E \left[\frac{s_k^2}{P_{k+1|k}} \big| e_{k|k-1} \right]
= \frac{K_k^2 E[u_k^2]}{P_{k+1|k}}.  \nonumber \\
\end{equation}
Substituting~\eqref{Eq:Upp_u},~\eqref{Eq:gain},~\eqref{Eq:PKC} and using Lemma~\ref{Lem:Upp_u} leads to
\begin{align}
E \left[\frac{s_k^2}{P_{k+1|k}} \big| e_{k|k-1} \right]
 &\le  \frac{e^2 \Lambda(\rho^*) \hat{z}_{k|k-1}}{\phi_{k} (1+\phi_{k}) P_{k|k-1}} \nonumber\\
 &\le  \frac{\hat{z}_{k|k-1}}{\phi_{k} (1+\phi_{k}) P_{k|k-1}},
\end{align}
which completes the proof.
\end{proof}

For simplification, we define $\xi_k$ as
\begin{equation}
\xi_k=  \frac{\hat{z}_{k|k-1}}{\phi_{k} (1+\phi_{k}) P_{k|k-1}}.
\label{Eq:xi_k}
\end{equation}

Armed with the above auxiliary lemmas, we next prove Theorem~\ref{Th:conditions}.

\begin{IEEEproof}[Proof of Theorem~\ref{Th:conditions}]
We use the following Lyapunov function to define the stochastic process:
\begin{equation}
V_k (e_{k|k-1}) = \frac{e_{k|k-1}^2}{P_{k|k-1}},
\end{equation}
which satisfies~\eqref{Eq:Tp_up} and \eqref{Eq:P0} as $P_{k|k-1}>0$.

We next use Lemma~\ref{Lem:Mstability} to develop the proof.
It follows from Lemma~\ref{Lem:bound_P} that the properties~\eqref{Eq:bou_v0} and \eqref{Eq:bou_vk} in Lemma~\ref{Lem:Mstability} are satisfied. Therefore, the main task left is to prove~\eqref{Eq:bou_Ev}.

From~\eqref{Eq:error_k+1}, we have
\begin{align}
V_{k+1}& (e_{k+1|k}) = \frac{e_{k+1|k}^2}{P_{k+1|k}}  = \frac{[(1- K_k C_k) e_{k|k-1} + s_k + m_k]^2}{P_{k+1|k}} \nonumber \\
=& \frac{(1-K_k C_k)^2}{P_{k+1|k}} e_{k|k-1}^2 + \frac{m_k [2(1-K_k C_k)e_{k|k-1}+m_k]}{P_{k+1|k}} \nonumber \\
 & + \frac{2s_k [(1-K_k C_k)e_{k|k-1}+m_k]}{P_{k+1|k}} + \frac{s_k^2}{P_{k+1|k}}.
\end{align}

By Lemmas~\ref{Lem:alpha},~\ref{Lem:mk} and~\ref{Lem:E_s2}, we have
\begin{multline}
E\left[V_{k+1} (e_{k+1|k}) |e_{k|k-1}\right]- V_{k}(e_{k|k-1})\\
\le  -\alpha_k V_{k}(e_{k|k-1})
     + b_k |e_{k|k-1}|^3
     + \xi_k.
\label{eq:aux1}
\end{multline}



We next proceed to bound the second term in $b_k$ in~\eqref{eq:aux1} as follows:
\begin{equation}
b_k |e_{k|k-1}|^3 \le \varsigma \alpha_k V_{k}(e_{k|k-1}),
\label{Eq:bound_b}
\end{equation}
where $0<\varsigma<1$ is preset controllable parameter. To prove the above inequality, we need to prove
\begin{equation}
|e_{k|k-1}| \le \frac{4\varsigma a_k^2 \phi_{k}\hat{z}_{k|k-1}}{4a_k \phi_{k}+|a'_k-1|}.
\end{equation}
Since $|e_{k|k-1}|=|a'_k-1|\hat{z}_{k|k-1}$, it suffices to show
\begin{equation}
|{a'_k}-1| \le \frac{4\varsigma a_k^2 \phi_{k}}{4{a_k} \phi_{k}+|a'_k-1|},
\label{Eq:cond_a}
\end{equation}
which is detailed as follows by distinguishing two cases:
\begin{itemize}
\item \textit{Case 1:} $0<{a'_k}\le 1$ when $a_k=a'_k$. In this case,~\eqref{Eq:cond_a} can be transformed into
    \begin{equation}
    (1-4 \phi_{k}-4 \phi_{k} \varsigma){a'_k}^2+(4\phi_{k}-2){a'_k}+1 \le 0
    \end{equation}
    With some algebraic operations, we obtain
        \begin{enumerate}
        \item $\frac{1-2\phi_{k}-2\sqrt{\phi_{k}(\phi_{k}+\varsigma)}}{1-4\phi_{k}(1+\varsigma)}<a'_k \le 1$, if $\phi_{k}<\frac{1}{4(1+\varsigma)}$,
        \item $\frac{2\phi_{k}-1+2\sqrt{\phi_{k}(\phi_{k}+\varsigma)}}{4\phi_{k}(1+\varsigma)-1} \le a'_k \le 1$, if $\phi_{k}>\frac{1}{4(1+\varsigma)}$,
        \item $\frac{1+\varsigma}{1+2\varsigma} \le a'_k \le 1$, if $\phi_{k}=\frac{1}{4(1+\varsigma)}$.
        \end{enumerate}
    Since it holds that $\frac{2\phi_{k}-1+2\sqrt{\phi_{k}(\phi_{k}+\varsigma)}}{4\phi_{k}(1+\varsigma)-1} <
      \frac{1+\varsigma}{1+2\varsigma}$ for every $\varsigma$ such that for $\phi_{k}\ge\frac{1}{4(1+\varsigma)}$,
    we have
        \begin{equation}
        \frac{1+\varsigma}{1+2\varsigma} \le a'_k \le 1.
        \end{equation}
\item \textit{Case 2:} $1<{a'_k}<2$ when $a_k=2-a'_k$. In this case,~\eqref{Eq:cond_a} can be transformed into
        \begin{multline}
        (1-4 \phi_{k}-4 \phi_{k} \varsigma){a'_k}^2+(12\phi_{k}+16\phi_{k} \varsigma -2){a'_k} \\
         +1-8\phi_{k}-16\phi_{k}\varsigma \le 0.
        \end{multline}
    With some algebraic operations, we get
        \begin{enumerate}
        \item $1<{a'_k}<\frac{1-6\phi_{k}-8\phi_{k}\varsigma
        +2\sqrt{\phi_{k}(\phi_{k}+\varsigma)}}{1-4\phi_{k}(1+\varsigma)}$, if $\phi_{k}<\frac{1}{4(1+\varsigma)}$,
        \item $1<{a'_k}<\frac{6\phi_{k}+8\phi_{k}\varsigma-1
        -2\sqrt{\phi_{k-1}(\phi_{k}+\varsigma)}}{4\phi_{k}(1+\varsigma)-1}$, if $\phi_{k}>\frac{1}{4(1+\varsigma)}$,
        \item $ 1< a'_k \le \frac{1+3\varsigma}{1+2\varsigma}$, if $\phi_{k}=\frac{1}{4(1+\varsigma)}$.
        \end{enumerate}
    Similar to Case 1, since $\frac{6\phi_{k}+8\phi_{k}\varsigma-1
        -2\sqrt{\phi_{k}(\phi_{k}+\varsigma)}}{4\phi_{k}(1+\varsigma)-1}>\frac{1+3\varsigma_k}{1+2\varsigma} $ if $\phi_{k}\ge\frac{1}{4(1+\varsigma)}$, we have
        \begin{equation}
        1< a'_k \le \frac{1+3\varsigma}{1+2\varsigma}.
        \end{equation}
\end{itemize}
It follows from the analysis of the two cases that if we set
\begin{equation}
\phi_{k} \ge \frac{1}{4(1+\varsigma)},
\label{Eq:phi_var}
\end{equation}
\eqref{Eq:cond_a} can be satisfied. Moreover, it holds that
\begin{equation}
|a'_k -1| \le \frac{\varsigma}{1+2\varsigma}.
\label{Eq:bou_a'}
\end{equation}
That is,
\begin{equation}
|e_{k|k-1}| \le \epsilon_k,
\label{Eq:epsilon_k}
\end{equation}
where $\epsilon_k \triangleq \frac{\varsigma}{1+2\varsigma} \hat{z}_{k|k-1}$.

By setting $\phi_k$ in~\eqref{Eq:phi_var}, we have
\begin{multline}
E\left[V_{k+1} (e_{k+1|k}) |e_{k|k-1}\right]- V_{k}(e_{k|k-1}) \\
\le  -(1-\varsigma)\alpha_k V_{k}(e_{k|k-1})
     + \xi_k,
\label{Eq:convergence_rate}
\end{multline}
for $|e_{k|k-1}| \le \epsilon_k$.

Therefore, we are able to apply Lemma~\ref{Lem:Mstability} to prove Theorem~\ref{Th:conditions} by setting $\epsilon = \frac{\varsigma}{1+2\varsigma} \hat{z}_{1|0}$, $\beta^* = \frac{1}{Q_0}$, $\alpha^*_k = (1-\varsigma)\alpha_k$, $\beta_k = \frac{1}{\overline p_k}$ and $\tau_k =\xi_k$.
\end{IEEEproof}

\begin{remark}
Theorem~\ref{Th:conditions} also holds in the sense of Lemma~\ref{Lem:stability} (the off-line version of Lemma~\ref{Lem:Mstability}) by setting the parameters in \eqref{Eq:stb_ineq} as
$\overline \beta = \frac{1}{Q_0}$,
$\alpha = \frac{1-\varsigma}{1+{\overline \phi}} \le \alpha^*_k $,
$\uline \beta = (P_{0|0} + Q_0 + {\overline q} ({\overline \phi} + 1) \ge \overline p_k$, and
$\tau = \frac{Q_0 \hat{z}_{max}}{\uline \phi (1+ \uline \phi)} \ge \xi_k $, where $\hat{z}_{max}$ is the maximum estimate.
\end{remark}

We conclude the analysis on the performance of our estimation algorithm for the static case with a more profound investigation on the evolution of the estimation error $|e_{k|k-1}|$. More specifically, we can distinguish three regions:
\begin{itemize}
\item \textit{Region 1:} $\sqrt{\frac{M \hat{z}_{k|k-1}}{\phi_{k}(M-1)(1-\varsigma)}} \le |e_{k|k-1}| \le \epsilon_k$. By substituting the condition into the right hand side of~\eqref{Eq:convergence_rate}, we obtain:
    \begin{equation*}
    -(1-\varsigma)\alpha_k V_{k}(e_{k|k-1}) + \xi_k
    \le  -\frac{(1-\varsigma)\alpha_k}{M} V_{k}(e_{k|k-1}),
    \end{equation*}
    where $M>1$ is a positive constant and can be set beforehand. It then follows that
    \begin{equation*}
    E\left[V_{k+1} (e_{k+1|k}) |e_{k|k-1}\right]
    \le \left(1- \frac{(1-\varsigma)\alpha_k}{M} \right) V_{k}(e_{k|k-1}).
    \end{equation*}
    Consequently, we can bound $E[e_{k|k-1}^2]$ as:
    \begin{equation}
    E[e_{k|k-1}^2] \le  \frac{\overline p_k}{Q_0} E[{e_{1|0}}^2]\prod_{i=1}^{k-1} {(1-\alpha^*_{i})}
    \label{Eq:exp}
    \end{equation}
    with $\alpha^*_k = \frac{(1-\varsigma)\alpha_k}{M}$. It can then be noted that $E[e_{k|k-1}^2] \to 0$ at an exponential rate as $k \to \infty$.
\item \textit{Region 2:} $\sqrt{\frac{\hat{z}_{k|k-1}}{\phi_{k}(1-\varsigma)}} \le |e_{k|k-1}| < \sqrt{\frac{M \hat{z}_{k|k-1}}{\phi_{k}(M-1)(1-\varsigma)}}$. In this case, we have
    \begin{equation*}
    -\frac{(1-\varsigma)\alpha_k}{M} V_{k}(e_{k|k-1})
    <-(1-\varsigma)\alpha_k V_{k}(e_{k|k-1}) + \xi_k  \le 0.
    \end{equation*}
    It then follows from Lemma~\ref{Lem:Mstability} that
    \begin{multline}
    E[e_{k|k-1}^2] \le \frac{\overline p_k}{Q_0} E[{e_{1|0}}^2]\prod_{i=1}^{k-1} {(1-\alpha^*_{i})} \\
    + \overline p_k \sum_{i=1}^{k-2} \xi_{k-i-1}\prod_{j=1}^{i}{(1-\alpha^*_{k-j})}.
    \end{multline}
    Hence, when $k \to \infty$, $E[e_{k|k-1}^2]$ converges at exponential rate to $\overline p_k \sum_{i=1}^{k-2} \xi_{k-i-1}\prod_{j=1}^{i}{(1-\alpha^*_{k-j})}\sim \Theta(\hat{z}_{k|k-1})$, which is decoupled with the initial estimation error and it thus holds $\displaystyle\frac{E[e_{k|k-1}]}{z_k} = \Theta \Big(\frac{1}{\sqrt{z_k}}\Big)\to 0$ when $z_k \to \infty$.

\item \textit{Region 3:} $0 \le |e_{k|k-1}| < \sqrt{\frac{\hat{z}_{k|k-1}}{\phi_{k}(1-\varsigma)}}$. In this case, we can show that the right hand side of~\eqref{Eq:convergence_rate} is positive, i.e.,
    \begin{equation*}
    -(1-\varsigma)\alpha_k V_{k}(e_{k|k-1}) + \xi_k >0.
    \end{equation*}
    It also follows from Lemma~\ref{Lem:Mstability} that
    \begin{multline}
    E[e_{k|k-1}^2] \le \frac{\overline p_k}{Q_0} E[{e_{1|0}}^2]\prod_{i=1}^{k-1} {(1-\alpha^*_{i})} \\
    + \overline p_k \sum_{i=1}^{k-2} \xi_{k-i-1}\prod_{j=1}^{i}{(1-\alpha^*_{k-j})}.
    \end{multline}
    Hence, when $k \to \infty$, $E[e_{k|k-1}^2]$ converges exponentially to $\overline p_k \sum_{i=1}^{k-2} \xi_{k-i-1}\prod_{j=1}^{i}{(1-\alpha^*_{k-j})} \sim \Theta(\hat{z}_{k|k-1})$, which is decoupled with the initial estimation error and it thus holds $\displaystyle\frac{E[e_{k|k-1}]}{z_k} \le \Theta \Big(\frac{1}{\sqrt{z_k}}\Big)\to 0$ when $z_k \to \infty$.
\end{itemize}
Combining the above three regions, we get the following results on the convergence of the expected estimation error $E[e_{k|k-1}]$: (1) if the estimation error is small (Region 3), it will converge to a value smaller than $\Theta(\sqrt {\hat{z}_{k|k-1}})$ as analysed in Region 3; (2) if the estimation error is larger (Region 1), it will decrease as analysed in Region 1 and fall into either Region 2 or Region 3 where $E[e_{k|k-1}] \le \Theta(\sqrt{\hat{z}_{k|k-1}})$ such that the relative estimation error $\frac{E[e_{k|k-1}]}{{z}_{k}} \to 0$ when  ${z}_{k}\to\infty$.

\subsection{Dynamic Case}

Our analysis on the stability of Algorithm~\ref{Al:SEKF} for the dynamic case is also composed of two steps. First, we derive the estimation error. Second, we establish the stability and the accuracy of Algorithm~\ref{Al:SEKF} in terms of the boundedness of estimation error.


We first derive the dynamics of the estimation error as follows:
\begin{align}
e_{k+1|k}
= (1- K_k C_k) e_{k|k-1} + s_k + m_k,
\label{Eq:error_k+1-w}
\end{align}
which differs from the static case~\eqref{Eq:error_k+1} in $s_k$. In the dynamic case, we have
\begin{align}
s_k &= w_k - K_k u_k
\label{Eq:sk-w}
\end{align}

Next, we show the boundedness of the estimation error in Theorem~\ref{Th:conditions-w}.

\begin{theorem}
Under the conditions of Algorithm~\ref{Al:SEKF-in}, consider the discrete-time stochastic system given by~\eqref{Eq:system_state} and~\eqref{Eq:measurement} and Algorithm~\ref{Al:SEKF}, if there exist time-varying positive real number $\lambda_k$, $\sigma_k>0$ such that
\begin{align}
    E[w_k] &\le  \lambda_k, \\
    E[{w_k}^2] &\le \sigma_k,
\end{align}
then the estimation error $e_{k|k-1}$ defined by~\eqref{Eq:error} is exponentially bounded in mean square and bounded w.p.o..
\label{Th:conditions-w}
\end{theorem}

\begin{remark}
Note that the condition $E[w_k] \le \lambda_k$ always holds for $E[w_k]<0$, we thus focus on the case that $E[w_k]\ge0$. In the proof, the explicit formulas of $\lambda_k$ and $\sigma_k$ are derived. As in the static case, the conditions may be too stringent such that the results still hold even if the conditions are not satisfied, as illustrated in the simulations.
\end{remark}

The proof of Theorem~\ref{Th:conditions-w} is also based on Lemmas~\ref{Lem:alpha},~\ref{Lem:mk} and~\ref{Lem:E_s2}, but due to the introduction of $w_k$ into $s_k$, we need another two auxiliary lemmas on $E[s_k]$ and $E[s_k^2]$.

\begin{lemma}
If $E[w_k] \ge 0$, then there exists a time-varying real number $d_k>0$ such that
\begin{multline}
E \left[\frac{2s_k [(1-K_k C_k)e_{k|k-1}+m_k]}{P_{k+1|k}} \Big| e_{k|k-1} \right] \\
    \le d_k |e_{k|k-1}| E[w_k].
\end{multline}
\label{Lem:E_w}
\end{lemma}

\begin{proof}
When $E[w_k] \ge 0$, from $E[v_k]=0$,~\eqref{Eq:a'},~\eqref{Eq:PKC} and the independence between $w_k$ and $e_{k|k-1}$, we can derive
\begin{align}
&E \left[\frac{2s_k [(1-K_k C_k)e_{k|k-1}+m_k]}{P_{k+1|k}} \Big| e_{k|k-1} \right]  \nonumber\\
&\le  2 E[w_k] \frac{1+\phi_{k}}{\phi_{k} P_{k|k-1}} \left[ \frac{\phi_{k} |e_{k|k-1}|}{1+\phi_{k}}
          + \frac{|e_{k|k-1}|^2}{2a(1+\phi_{k}) \hat z_{k|k-1}}\right] \nonumber\\
&\le  E[w_k] \frac{2a_k \phi_{k}+|a'_k-1|}{a_k \phi_{k}P_{k|k-1}} |e_{k|k-1}|.
\end{align}

We thus complete the proof by setting
\begin{equation}
d_k = \frac{2a_k\phi_{k}+|a'_k-1|}{a_k \phi_{k}P_{k|k-1}}.
\label{Eq:dk}
\end{equation}
\end{proof}

\begin{lemma}
There exists a time-varying parameter $\xi_k^*>0$ such that
$E \left[\frac{{s_k}^2}{P_{k+1|k}} \big| e_{k|k-1} \right] \le \xi_k^*$.
\label{Lem:E_w2}
\end{lemma}

\begin{proof}
By~\eqref{Eq:sk-w}, we have
\begin{equation}
s_k^2 = w_k^2 - 2K_k w_k u_k + K_k^2 u_k^2.
\end{equation}
Since $w_k$ and $u_k$ are uncorrelated and $e_{k|k-1}$ does not depend on either $w_k$ or $u_k$, we have
\begin{equation}
E \left[\frac{s_k^2}{P_{k+1|k}} \big| e_{k|k-1} \right]
= \frac{E[w_k^2]}{P_{k+1|k}} + \frac{K_k^2 E[u_k^2]}{P_{k+1|k}}.
\label{Eq:w_u}
\end{equation}
Substituting~\eqref{Eq:Upp_u},~\eqref{Eq:gain},~\eqref{Eq:PKC} and using Lemma~\ref{Lem:Upp_u}, noticing that $E[u_k]=0$, we get
\begin{multline}
E \left[\frac{s_k^2}{P_{k+1|k}} \big| e_{k|k-1} \right]
\le \frac{1+\phi_{k}}{\phi_{k}P_{k|k-1}}E[w_k^2]
    + \frac{e^2 \Lambda(\rho^*)\hat{z}_{k|k-1}}{\phi_{k} (1+\phi_{k}) P_{k|k-1}} \\
\le \frac{1+\phi_{k}}{\phi_{k}P_{k|k-1}}E[w_k^2]
    + \frac{\hat{z}_{k|k-1}}{\phi_{k} (1+\phi_{k}) P_{k|k-1}}.
\end{multline}
Finally, by setting $\xi_k^*$ as
\begin{equation}
\xi_k^*= \frac{1+\phi_{k}}{\phi_{k}P_{k|k-1}}E[{w_k}^2]
       + \frac{\hat{z}_{k|k-1}}{\phi_{k} (1+\phi_{k}) P_{k|k-1}},
\label{Eq:xi_k-w}
\end{equation}
we complete the proof.
\end{proof}

Armed with the above lemmas, we next prove Theorem~\ref{Th:conditions-w} by utilizing the same method with the proof of Theorem~\ref{Th:conditions}.

\begin{IEEEproof}[Proof of Theorem~\ref{Th:conditions-w}]
Recall~\eqref{Eq:sk} and~\eqref{Eq:sk-w}, we notice that the only difference between the estimation errors of Algorithms~\ref{Al:SEKF} and~\ref{Al:SEKF-in} is $s_k$. Therefore, it suffices to study the impact of $w_k$ on $V_k (e_{k|k-1})$.

It follows from Lemmas~\ref{Lem:alpha},~\ref{Lem:mk},~\ref{Lem:E_s2},~\ref{Lem:E_w} and~\ref{Lem:E_w2} that
\begin{multline}
E\left[V_{k+1} (e_{k+1|k}) |e_{k|k-1}\right]- V_{k}(e_{k|k-1})
\le  -\alpha_k V_{k}(e_{k|k-1})\\
     + b_k |e_{k|k-1}|^3
     + d_k |e_{k|k-1}| E[w_k]
     + \xi^*_k.
\end{multline}

Furthermore, bounding the second item in $b_k$ as~\eqref{Eq:bound_b} and given $\phi_k$ in~\eqref{Eq:phi_var}, yields
\begin{multline}
E\left[V_{k+1} (e_{k+1|k}) |e_{k|k-1}\right]- V_{k}(e_{k|k-1}) \\
\le  -(1-\varsigma)\alpha_k V_{k}(e_{k|k-1})
     + d_k |e_{k|k-1}| E[w_k]
     + \xi^*_k
\end{multline}
for $|e_{k|k-1}| \le \epsilon_k$.

And we can thus prove Theorem~\ref{Th:conditions-w} by setting $\epsilon = \frac{\varsigma}{1+2\varsigma} \hat{z}_{1|0}$, $\beta^* = \frac{1}{Q_0}$, $\alpha^*_k = (1-\varsigma)\alpha_k$, $\tau_k =\xi^*_k + d_k |e_{k|k-1}| \lambda_k$ and $\beta_k = \frac{1}{\overline p_k}$.
\end{IEEEproof}

We conclude the analysis on the performance of our estimation algorithm for the dynamic case with a more profound investigation on the evolution of the estimation error $|e_{k|k-1}|$ and derive the explicit formulas for $\lambda_k$ and $\sigma_k$. More specifically, we can distinguish three regions:
\begin{itemize}
\item \textit{Region 1}: $\sqrt{\frac{4M\hat{z}_{k|k-1}}{\phi_{k}(M-1)(1-\varsigma)}} \le |e_{k|k-1}| \le \epsilon_k$. In this case, the objective is to achieve
     \begin{multline}
          E\left[V_{k+1} (e_{k+1|k}) |e_{k|k-1}\right]- V_{k}(e_{k|k-1}) \\
           \le  -\frac{1}{M} (1-\varsigma)\alpha_k V_{k}(e_{k|k-1})
           \label{Eq:super_large}
      \end{multline}
      so that $E[e_{k|k-1}^2]$ is bounded as
       \begin{equation}
    E[e_{k|k-1}^2] \le  \frac{\overline p_k}{Q_0} E[{e_{1|0}}^2]\prod_{i=1}^{k-1} {(1-\alpha^*_{i})}.
    \label{Eq:exp-w}
    \end{equation}
That is, it should hold that
\begin{equation*}
d_k |e_{k|k-1}| E[w_k] + \xi_k^* \le \frac{M-1}{M} (1-\varsigma)\alpha_k V_{k}(e_{k|k-1}).
\end{equation*}
To that end, we firstly let the following inequalities hold
\begin{eqnarray}
\begin{cases}
d_k |e_{k|k-1}| E[w_k] \le \frac{M-1}{2M} (1-\varsigma)\alpha_k V_{k}(e_{k|k-1}), \\
\xi_k^* \le \frac{M-1}{2M} (1-\varsigma)\alpha_k V_{k}(e_{k|k-1}).
\end{cases}
\label{Eq:super_conditions}
\end{eqnarray}

Secondly, substituting~\eqref{Eq:dk},~\eqref{Eq:xi_k-w} into~\eqref{Eq:super_conditions} leads to
\begin{equation}
E[w_k] \le  \frac{a_k \phi_{k} (1-\varsigma) |e_{k|k-1}|}
                  {(1+\phi_{k})\left(2a_k\phi_{k}+|a'_k-1|\right)},
\label{Eq:Ew}
\end{equation}
\begin{eqnarray}
E[{w_k}^2] \le \frac{\phi_{k}(M-1)(1-\varsigma){|e_{k|k-1}|}^2- 2M\hat{z}_{k|k-1}}{2M(1+\phi_{k})^2}. 
\label{Eq:Ew2}
\end{eqnarray}
Thirdly, let
\begin{equation}
\frac{\phi_{k}(M-1)(1-\varsigma){|e_{k|k-1}|}^2}{2M(1+\phi_{k})^2}
                \ge \frac{2\hat{z}_{k|k-1}}{(1+\phi_{k})^2},
\label{Eq:w2-u2}
\end{equation}
and we thus have
\begin{align}
|e_{k|k-1}| &\ge \sqrt{\frac{4M\hat{z}_{k|k-1}}{\phi_{k}(M-1)(1-\varsigma)}} \triangleq \tilde \epsilon, \\
E[{w_k}^2]  &\le \frac{ \hat{z}_{k|k-1}}{(1+\phi_{k})^2} \triangleq \sigma_k. \label{Eq:sigma}
\end{align}
The rational behind can be interpreted as follows: i) the right term of~\eqref{Eq:Ew2} cannot be less than zero and ii) there always exists the measurement uncertainty in the system. Consequently, the impact of tag population change plus the measurement uncertainty should equal in order of magnitude that of only measurement uncertainty, which can be achieved by establishing $E[{w_k}^2] \le K_k^2 E[{u_k}^2]$ and \eqref{Eq:w2-u2} with reference to~\eqref{Eq:w_u} and~\eqref{Eq:xi_k-w}. \\

However, since $a'_k$ and $a_k$ are unknown a priori, we thus need to transform the right hand side of~\eqref{Eq:Ew} to a computable form. From~\eqref{Eq:aa'} and~\eqref{Eq:bou_a'}, we get
\begin{eqnarray*}
\begin{cases}
|a'_k - 1| = |a_k - 1| \\
|1 - \frac{1}{a_k}| \le \frac{\varsigma}{1+3\varsigma}
\end{cases}
\end{eqnarray*}
such that it holds for the right hand side of~\eqref{Eq:Ew} that
\begin{eqnarray*}
\frac{a_k \phi_{k} (1-\varsigma) |e_{k|k-1}|}
                  {3(1+\phi_{k})\left(2a_k\phi_{k}+|a'_k-1|\right)}
\ge \frac{\phi_{k} (1-\varsigma) \tilde \epsilon}
                  {3(1+\phi_{k})\left(2\phi_{k}+\frac{\varsigma}{1+3\varsigma}\right)}.
\end{eqnarray*}

Finally, let
\begin{eqnarray}
E[w_k]
\le \frac{\phi_{k} (1-\varsigma)  \tilde \epsilon}
                  {3(1+\phi_{k})\left(2\phi_{k}+\frac{\varsigma}{1+3\varsigma}\right)} \triangleq \lambda_k,
\label{Eq:lambda}
\end{eqnarray}

we can establish \eqref{Eq:exp-w} and thus get that $E[e_{k|k-1}^2] \to 0$ at an exponential rate when $k \to \infty$.

\item \textit{Region 2:} $\sqrt{\frac{4\hat{z}_{k|k-1}}{\phi_{k}(1-\varsigma)}} \le |e_{k|k-1}| < \sqrt{\frac{4M\hat{z}_{k|k-1}}{\phi_{k}(M-1)(1-\varsigma)}}$. Given $\tilde \epsilon$, $\lambda_k$ and $\sigma_k$ as in \textit{Region 1}, in this case, we have
  \begin{equation*}
    -(1-\varsigma)\alpha_k V_{k}(e_{k|k-1}) + d_k |e_{k|k-1}| E[w_k]+ \xi_k^* \le 0.
  \end{equation*}
   It then follows from Lemma~\ref{Lem:Mstability} that
    \begin{multline}
    E[e_{k|k-1}^2] \le \frac{\overline p_k}{Q_0} E[{e_{1|0}}^2]\prod_{i=1}^{k-1} {(1-\alpha^*_{i})} \\
    + \overline p_k \sum_{i=1}^{k-2} \tau_{k-i-1}\prod_{j=1}^{i}{(1-\alpha^*_{k-j})}.
    \end{multline}
    Hence, when $k \to \infty$, $E[e_{k|k-1}^2]$ converges exponentially to $\overline p_k \sum_{i=1}^{k-2} \tau_{k-i-1}\prod_{j=1}^{i}{(1-\alpha^*_{k-j})}\sim \Theta(\hat{z}_{k|k-1})$ and it thus holds that $\frac{E[e_{k|k-1}]}{{z}_{k}} = \Theta(\frac{1}{\sqrt{{z}_{k}}}) \to 0$ for $z_k \to \infty$.
\item \textit{Region 3:} $0 \le |e_{k|k-1}| < \sqrt{\frac{4\hat{z}_{k|k-1}}{\phi_{k}(1-\varsigma)}}$. The circumstances in this region are very complicated due to $E[w_k]$ and $E[{w_k}^2]$, we here thus just consider the worst case that $E[w_k]=\lambda_k$ and $E[{w_k}^2]=\sigma_k$. Consequently, we have
    \begin{equation*}
    -(1-\varsigma)\alpha_k V_{k}(e_{k|k-1}) + d_k |e_{k|k-1}| E[w_k]+ \xi_k^* > 0,
  \end{equation*}
  and it then follows from Lemma~\ref{Lem:Mstability} that
    \begin{multline}
    E[e_{k|k-1}^2] \le \frac{\overline p_k}{Q_0} E[{e_{1|0}}^2]\prod_{i=1}^{k-1} {(1-\alpha^*_{i})} \\
    + \overline p_k \sum_{i=1}^{k-2} \tau_{k-i-1}\prod_{j=1}^{i}{(1-\alpha^*_{k-j})}.
    \end{multline}
  Hence, when $k \to \infty$, $E[e_{k|k-1}^2]$ converges at exponential rate to $\overline p_k \sum_{i=1}^{k-2} \tau_{k-i-1}\prod_{j=1}^{i}{(1-\alpha^*_{k-j})}\sim \Theta(\hat{z}_{k|k-1})$, and thus $\frac{E[e_{k|k-1}]}{{z}_{k}} \le \Theta(\frac{1}{\sqrt{{z}_{k}}}) \to 0$ for $z_k \to \infty$.\\

  Note that for the case that $E[w_k]<\lambda_k$ and $E[{w_k}^2]<\sigma_k$, the range of \textit{Region 3} will shrink and the range of \textit{Region 2} will largen.
\end{itemize}

Integrating the above three regions, we can get the similar results on the convergence of the expected estimation error $E[e_{k|k-1}]$ as in the static case.

\section{Numerical Analysis}
\label{sec:evaluation}

In this section, we conduct extensive simulations to evaluate the performance of the proposed tag population estimation algorithms by focusing on the relative estimation error denoted as $\left| \frac{z_k - \hat z_{k|k-1}}{z_k} \right|$. Specifically, we simulate in sequence both static and dynamic RFID systems in two scenarios where the initial tag population are $z_0=10^4$ (scenario 1) and $z_0=10^5$ (scenario 2) with the following parameters: $q=0.1$, $P_{0|0} =1$, $J=3$, $\theta=4$ and $\Upsilon=0.5$ with reference to~\eqref{Eq:theta} and~\eqref{Eq:Upsilon}, $\uline \phi = 0.25$ and $\overline \phi = 10$ such that~\eqref{Eq:phi_var} always holds. Since the proposed algorithms do not require collision detection, we set a slot to $0.4$ms as in the EPCglobal C1G2 standard~\cite{C1G22005}.

\subsection{Static System}

We evaluate the performance by varying initial relative estimation error as
\begin{itemize}
\item $\big|\frac{z_0-\hat z_{0|0}}{z_0}\big|=0.9$ means a large initial estimation error.
\item $\big|\frac{z_0-\hat z_{0|0}}{z_0}\big|=0.5$ means a medium initial estimation error.
\item $\big|\frac{z_0-\hat z_{0|0}}{z_0}\big|=0.2$ implies a small initial estimation error and satisfies~\eqref{Eq:epsilon_k} with $0.5 \le \varsigma <1$, .
\end{itemize}
The purpose of the first two cases is to investigate the effectiveness of the estimation in relative large initial estimation errors while the third case intends to verify the analytical results $\frac{\hat z_{0|0}}{z_0} >0.75$ as shown in~\eqref{Eq:epsilon_k}. In this part, the simulation runs $10$ rounds each time.

\begin{figure}[htbp]
\centering
\includegraphics[width=8cm]{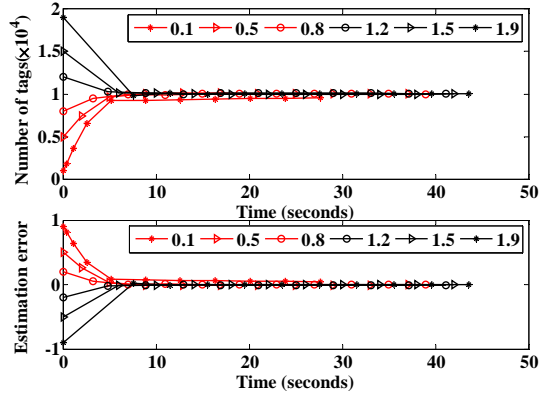}
\caption{Estimation for static tag population ($z_0=10^4$).}
\label{fig:1-C-4}
\end{figure}

\begin{figure}[htbp]
\centering
\includegraphics[width=8cm]{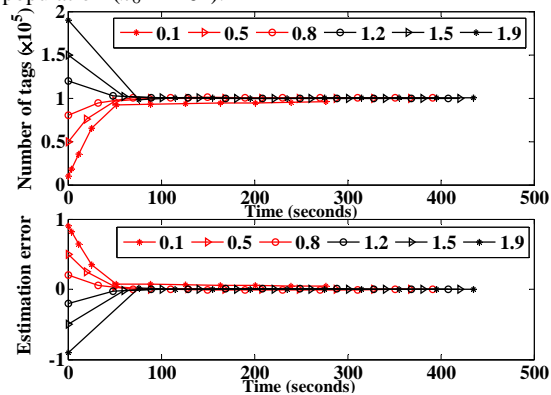}
\caption{Estimation for static tag population ($z_0=10^5$).}
\label{fig:1-C-5}
\end{figure}

Fig.~\ref{fig:1-C-4} and Fig.~\ref{fig:1-C-5} illustrates the estimation processes in the first two scenarios, with different initial estimation errors ($\hat z_{0|0}$ ranges from $1000$ to $19000$ with $z_0=10000$). As shown in the figures, the estimation $\hat z_{k|k-1}$ converges towards the actual number of tags within very short time in all the six cases, despite the initial estimation error. Moreover, the convergence time is shorter with smaller $\hat z_{0|0}$, i.e., under an under-estimation of the initial tag population. The reason is that from~\eqref{Eq:gain},~\eqref{Eq:dh} and~\eqref{Eq:Rk}, smaller $\phi_k$ and larger $\hat z_{k|k-1}$ lead to larger $K_k$, which increases the convergence rate.







\subsection{Dynamic system}

We simulate the case $\frac{z_0-\hat z_{0|0}}{z_0} = 0.5$ to evaluate the performance of the EKF-based estimator for dynamic systems. In the subsection, the tag population in scenario 1 varies in order of magnitude from $\sqrt {\hat z_{k|k-1}}$ to $0.4 \hat z_{k|k-1}$, while that in scenario 2 from $\sqrt {\hat z_{k|k-1}}$ to $0.5 \hat z_{k|k-1}$. Moreover, the tag population changes in different patterns during the simulation as shown in Fig.~\ref{fig:2-V-4} and Fig.~\ref{fig:2-V-5}.

\begin{figure}[htbp]
\centering
\includegraphics[width=8cm]{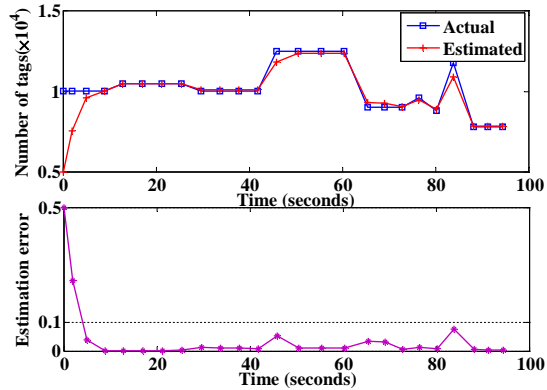}
\caption{Estimation for dynamic tag population ($z_0=10^4$).}
\label{fig:2-V-4}
\end{figure}

\begin{figure}[htbp]
\centering
\includegraphics[width=8cm]{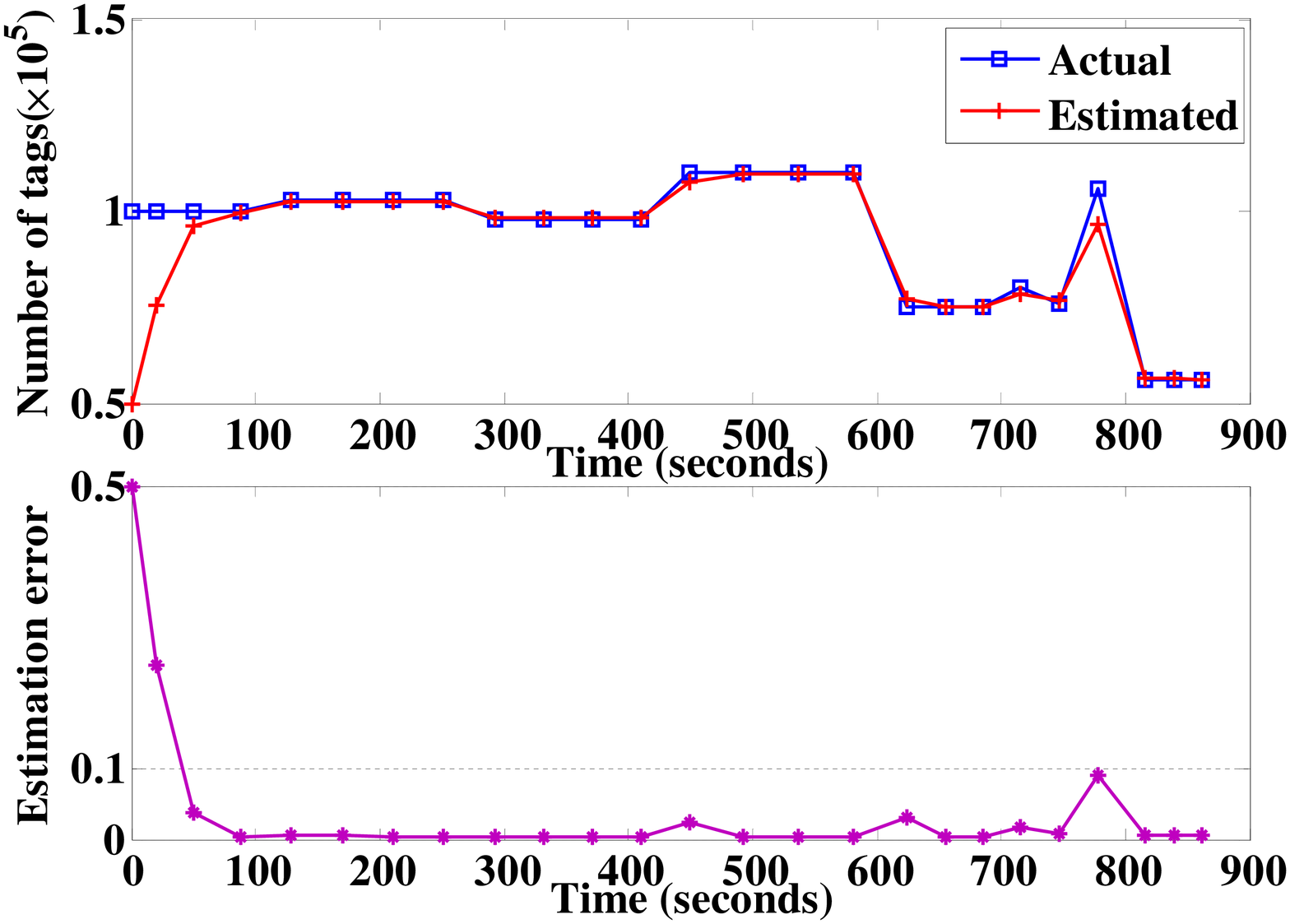}
\caption{Estimation for dynamic tag population ($z_0=10^5$).}
\label{fig:2-V-5}
\end{figure}
From Fig.~\ref{fig:2-V-4} and Fig.~\ref{fig:2-V-5}, we make the following observations. First, as derived in Theorem~\ref{Th:conditions-w}, the estimation is stable and accurate facing to a relative small population change, i.e., around the order of magnitude $\sqrt {\hat z_{k|k-1}}$. Second, the proposed scheme also functions nicely even when the initial estimation error is as high as $0.4 \hat z_{k|k-1}$ and $0.5 \hat z_{k|k-1}$ tags as shown in Fig.~\ref{fig:2-V-4} and Fig.~\ref{fig:2-V-5}, respectively. This is due to the CUSUM-based change detection which detects state changes promptly such that a small value is set for $\phi_k$, leading to rapid convergence rate. Third, it can be noted from the comparison between Fig.~\ref{fig:2-V-4} and Fig.~\ref{fig:2-V-5} that under the same load factor $\rho$, the estimation scheme used in the larger scale system is more accurate and stable. The main reason is that a frame in a large RFID system is much longer, which reduces the measurement variance $Var[u_k]$. Thus, the measurement is more accurate.


\section{Conclusion}
\label{sec:conclusion}

In this paper, we have addressed the problem of tag estimation in dynamic RFID systems and designed a generic framework of stable and accurate tag population estimation schemes based on Kalman filter.
Technically, we leveraged the techniques in extended Kalman filter (EKF) and cumulative sum control chart (CUSUM) to estimate tag population for both static and dynamic systems. By employing Lyapunov drift analysis, we mathematically characterised the performance of the proposed framework in terms of estimation accuracy and convergence speed by deriving the closed-form conditions on the design parameters under which our scheme can stabilise around the real population size with bounded relative estimation error that tends to zero within exponential convergence rate. In future work, we plan to use the theoretical framework developed in this work to address tag estimation problems with multiple readers with overlapping covered areas.

\bibliographystyle{abbrv}
\bibliography{reference_RFID_estimation}

\end{document}